\PassOptionsToPackage{bookmarks=true,colorlinks,citecolor=blue,linkcolor=blue,urlcolor=blue}{hyperref}
\documentclass[sigplan,10pt]{acmart}

\setcopyright{none}

\usepackage{amsmath,amsfonts}
\usepackage{algorithmicx}
\usepackage{algpseudocode}
\usepackage{algorithm}
\usepackage{textcomp}
\usepackage{comment}
\usepackage{wasysym}
\usepackage{enumitem}
\usepackage{bm}
\usepackage{cleveref}
\usepackage{makecell}
\usepackage{array}
\usepackage{wrapfig}
\usepackage{threeparttable}
\usepackage{amsthm}
\usepackage{subcaption}
\usepackage{pifont}

\crefname{section}{§}{§§}
\Crefname{section}{§}{§§}
\usepackage{graphicx}      
\usepackage{fontawesome5}  

\graphicspath{{images/}}

\definecolor{MyViolet}{RGB}{148,0,211}

\makeatletter
\patchcmd{\@mktitle@iii}{\par\bigskip}{\par\medskip}{}
{\PackageWarning{Poseidon}{Could not adjust title spacing}}
\patchcmd{\@mkauthors@iii}{\and\par\bigskip}{\and\par\medskip}{}
{\PackageWarning{Poseidon}{Could not adjust author spacing}}
\AtBeginMaketitle{%
	\if@ACM@anonymous
	\let\@mkauthors\@empty
	\fi
}
\author{%
	\vspace{0.75em}
	\normalsize
	\begin{tabular}{@{}ccc@{}}
		\begin{tabular}{@{}c@{}}
			Xiaosong Chen$^{*}$ \\
			University of Macau \\
			\texttt{yc27909@connect.um.edu.mo}
		\end{tabular}
		&
		\begin{tabular}{@{}c@{}}
			Shaoheng Nie$^{*}$ \\
			Fudan University \\
			\texttt{shnie23@m.fudan.edu.cn}
		\end{tabular}
		&
		\begin{tabular}{@{}c@{}}
			Zhongmin Zhao \\
			University of Macau \\
			\texttt{yc47494@connect.um.edu.mo}
		\end{tabular}
		\\[1.5em]
		\begin{tabular}{@{}c@{}}
			Zizhao Mo \\
			University of Macau \\
			\texttt{yc17461@connect.um.edu.mo}
		\end{tabular}
		&
		\begin{tabular}{@{}c@{}}
			Jiapeng Chen \\
			Fudan University \\
			\texttt{22300240004@m.fudan.edu.cn}
		\end{tabular}
		&
		\begin{tabular}{@{}c@{}}
			Huanle Xu$^{\dagger}$ \\
			University of Macau \\
			\texttt{huanlexu@um.edu.mo}
		\end{tabular}
		\\[1.5em]
		\begin{tabular}{@{}c@{}}
			Zeren Li$^{\dagger}$ \\
			Independent researcher \\
			\texttt{lzr010506@gmail.com}
		\end{tabular}
		&
		\begin{tabular}{@{}c@{}}
			Weiwei Sun$^{\dagger}$ \\
			Fudan University \\
			\texttt{wwsun@fudan.edu.cn}
		\end{tabular}
		&
		\begin{tabular}{@{}c@{}}
			ChengZhong Xu \\
			University of Macau \\
			\texttt{czxu@um.edu.mo}
		\end{tabular}
		\\[1.5em]
		\multicolumn{3}{c}{
			\footnotesize
			$^{*}$Equal contribution. \quad
			$^{\dagger}$Corresponding authors.
		}
	\end{tabular}
	\vspace{0.75em}
}
\renewcommand{\authors}{Xiaosong Chen \and Shaoheng Nie \and Zhongmin Zhao \and Zizhao Mo \and Jiapeng Chen \and Huanle Xu \and Zeren Li \and Weiwei Sun \and ChengZhong Xu}
\renewcommand{\shortauthors}{Xiaosong Chen et al.}
\AtBeginDocument{%
	\fancyhead{}%
	\fancyhead[L]{\ACM@linecountL}%
	\fancyhead[R]{\ACM@linecountR}%
}
\makeatother

\newtheoremstyle{mytheoremstyle}
{1pt}{1pt}{\itshape}{}{\bfseries}{.}{ }
{\thmname{#1}\thmnumber{ #2}\thmnote{ (#3)}}
\theoremstyle{mytheoremstyle}
\newtheorem{theorem}{Theorem}
\newtheorem{assumption}{Assumption}
\newtheorem{definition}{Definition}
\newtheorem{lemma}{Lemma}

\def\BibTeX{{\rm B\kern-.05em{\sc i\kern-.025em b}\kern-.08em
		T\kern-.1667em\lower.7ex\hbox{E}\kern-.125emX}}

\renewcommand\footnotetextcopyrightpermission[1]{}

\begin{document}
	
	
	\pagenumbering{arabic}

	\title{Poseidon: DAG-Guided Parallelism Search for LLM Pre-Training on Heterogeneous Clusters}

	\begin{abstract}
		
		
		With the rapid advancement of accelerator technologies, pre-training large language models (LLMs) on heterogeneous accelerator clusters has become increasingly crucial for maximizing hardware utilization. Existing systems, however, suffer from inaccurate training time modeling, which undermines the parallelization optimizations built upon it. Moreover, for current approaches, the vast configuration search space makes exhaustive exploration infeasible, forcing a trade-off between search time and training efficiency.
		
		To overcome these limitations, we introduce Poseidon, an efficient and scalable LLM training framework designed with heterogeneity awareness. Its core is an explicit training time model based on a directed acyclic graph. Building on this graph, Poseidon employs two efficient, theoretically grounded strategies: stage-level pruning via early stopping with partial estimation, and layer-to-stage mapping exploiting a ridge-like distribution pattern. These strategies reduce the search space without sacrificing optimal training efficiency. Experiments on heterogeneous clusters show that Poseidon improves training throughput by up to $2.76\times$ over state-of-the-art systems. 
	\end{abstract}
	
	
	\maketitle

	\section{Introduction}
	LLMs~\cite{deepseek, touvron2023llama, gpt} have demonstrated unprecedented capabilities, advancing applications such as code completion~\cite{touvron2023llama2, deepseekcoder}, story writing~\cite{storywriting}, and conversational agents~\cite{lmsyschat}. However, the scaling law drives the development of increasingly larger foundational models to fully exploit their potential, requiring substantial computational power. For instance, training LLaMA-3 requires $3.8\times10^{25}$ FLOPs, as reported by Meta~\cite{grattafiori2024llama}. This has led to widespread deployment of accelerators such as GPUs~\cite{mittalV2019gpu} and NPUs~\cite{chen2020npu} to support these resource-intensive workloads. Meanwhile, accelerators evolve through multiple generations, resulting in notable inter-generation heterogeneity. Computational power and memory can vary significantly within the same accelerator class, leading to diverse training performance. For example, limited memory constrains batch sizes on older accelerators, while differences in computational power create significant gaps in training time for identical workloads.
	
	To optimize cluster resource utilization, pre-training LLMs on heterogeneous accelerator devices has become increasingly important. Several studies have explored this problem~\cite{jia2022whale, yan2024flashflex, um2024metis}, focusing on optimizing parallelization configurations to minimize training time on such resources. These configurations typically combine pipeline, tensor, and data parallelism to improve computational efficiency, reduce memory usage, and lower communication overhead.
	
	However, accelerator heterogeneity introduces significant complexity in identifying optimal parallelization configurations. This is because configuration searching is exposed to a new dimension, where the solution space grows exponentially with the number of accelerator types involved. In witnessing this dilemma, existing works propose heuristic methods to balance searching time and training throughput. Although they reduce the searching time to some extent, we identify three key obstacles that can lead to low training efficiency. First, the cost models utilized by existing works to capture the training time are inaccurate. Since they cannot precisely characterize the overlap among concurrently proceeding operations, we demonstrate up to $1.32\times$ deviation in our experimental testbed. Second, the heuristic pruning strategies, which serve as the key to reducing searching time, risk neglecting the optimal configuration. This is because the rigid strategies proposed cannot flexibly adapt to various cluster environments, where inter-accelerator connections and computational power gaps can be highly diverse. In this sense, we witness up to $1.91\times$ degradation when evaluating the pruning strategies for baseline systems. Third, existing modeling techniques present low generalizability, due to the tight coupling with a specific training mode, i.e., 1F1B \cite{narayanan2019pipedream}, where each pipeline stage alternates between executing one forward pass and one backward pass of micro-batches after pipeline warm-up. As a result, popular training variants, including Eager 1F1B \cite{zhuang2023optimizing} and Interleaved 1F1B \cite{narayanan2021efficient}, cannot be well supported by this design.
	
	In this paper, we propose Poseidon, a novel LLM pre-training system designed to address the aforementioned limitations.  Poseidon efficiently explores optimal parallel configurations across heterogeneous accelerator resources while maintaining low search overhead. To achieve this, it incorporates an explicit directed acyclic graph (DAG)-based training time estimation model that captures stage- and batch-level execution dependencies. This graph-based representation models the data execution order within each stage and the dependencies across micro-batches for arbitrary parallel training configurations, using distinct node and edge types. It enables training time estimation over heterogeneous resources by computing the longest path. Furthermore, the graph accounts for communication-computation overlap in heterogeneous environments, significantly enhancing the accuracy of training time predictions.
	
	Leveraging this estimation framework, Poseidon integrates two optimization strategies to efficiently navigate the vast configuration search space. First, it employs stage-level pruning with early stopping, guided by a custom lower-bound function. This mechanism terminates the evaluation of suboptimal configurations once partial stage analysis reveals inefficiencies, avoiding exhaustive search. Second, Poseidon exploits a ridge-like distribution pattern in layer-to-stage mapping to prune suboptimal regions of the search space. This approach, backed by theoretical guarantees, substantially accelerates search while guaranteeing optimal solutions. Together, these strategies enable Poseidon to achieve unprecedented efficiency in identifying optimal parallel configurations for LLM training over heterogeneous clusters. In summary, we have made the following contributions: 
	
	\scalebox{0.65}{\faPlane} We introduce a novel DAG-based cost model for LLM training, capturing various  factors and therefore enabling precise training time estimation.  This model also offers high generalizability to multiple training variants.
	
	\scalebox{0.65}{\faPlane} We propose two pruning strategies for configuration exploration that are theoretically guaranteed to identify the optimal parallel configuration.
	
	\scalebox{0.65}{\faPlane} We evaluate Poseidon via extensive experiments on heterogeneous clusters. Results show that it consistently outperforms baselines, achieving up to {$2.76\times$} higher throughput.
	
	\begin{figure}[t]
		\centering
		\includegraphics[width=0.49\textwidth]{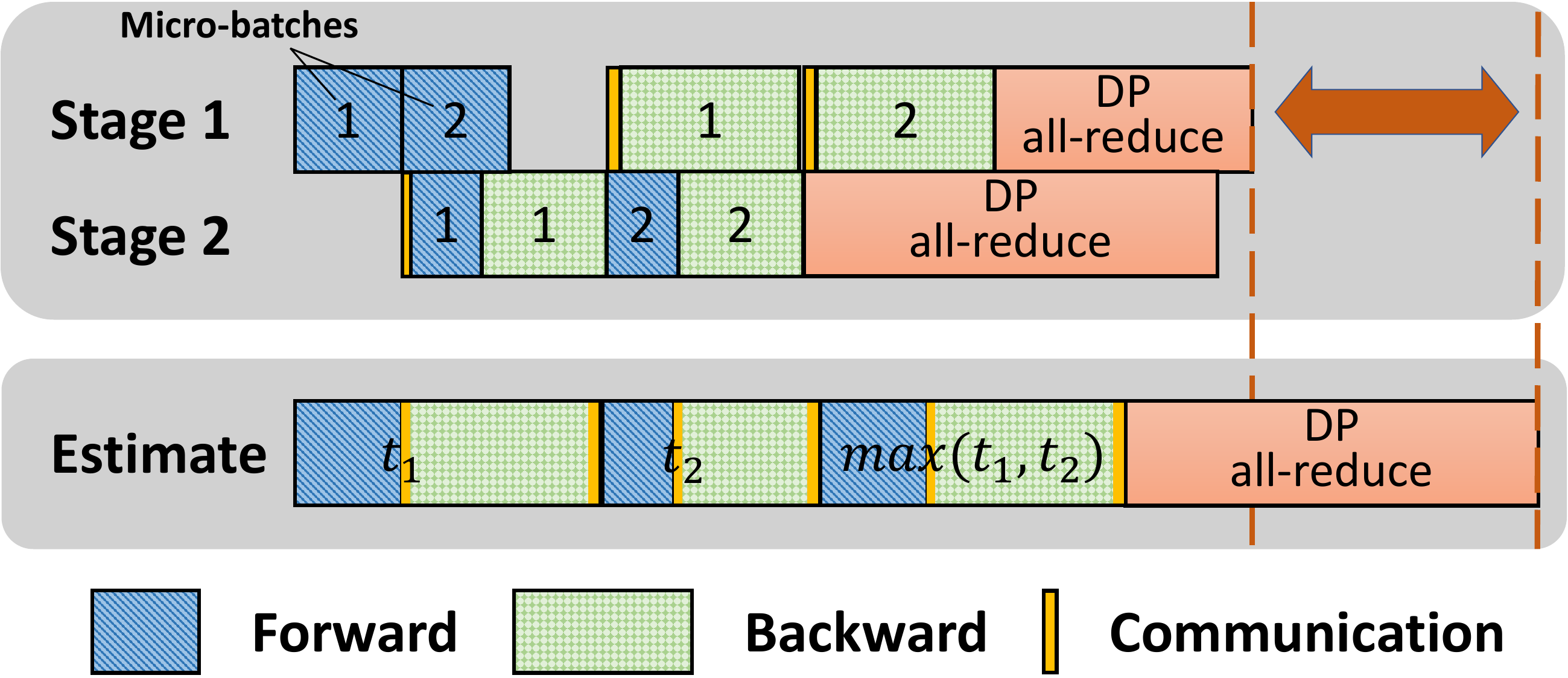}
		\vspace{-1.75em}
		\caption{An illustrative example to demonstrate why the cost model adopted by existing works cannot accurately characterize the training time. Here, $t_1$ and $t_2$ denote the total computation and communication costs for each micro-batch in stages 1 and 2, respectively, as defined in Eq.~\eqref{eq:static_cost_model}.}
		\vspace{-1.5em}
		\label{fig:deviation_cost_model}
	\end{figure}
	
	\vspace{-0.5em}
	\section{Background and Motivation}
	
	\subsection{LLM Training under Accelerator Heterogeneity}
	Training an LLM involves processing a batch of data through a sequence of layers during the forward pass (FP) and backward pass (BP) to compute gradients. As models scale, distributed training partitions parameters and data across devices to alleviate memory bottlenecks and accelerate computation~\cite{gusak2022survey}. Common strategies include pipeline parallelism (PP)~\cite{narayanan2021memory, huang2019gpipe, narayanan2019pipedream}, tensor parallelism (TP)~\cite{narayanan2021efficient, shoeybi2019megatron}, and data parallelism (DP)~\cite{li2020pytorch, dean2012large}, dividing computation at layer, operator, and data levels, respectively.
	
	Meanwhile, modern clusters increasingly contain accelerator generations with different compute power, memory capacities, and interconnect bandwidths~\cite{ye2024deep}. Such heterogeneity complicates LLM parallelization because synchronous operations, such as parameter updates, are bottlenecked by the slowest device~\cite{vaswani2017attention, brown2020language}. Recent work therefore seeks to balance computation across devices while reducing communication overhead~\cite{um2024metis, yan2024flashflex}.


	\subsection{Limitations of Existing Pre-training Systems}
	\label{sec:limitation_of-existing_works}
	
	Several heterogeneity-aware LLM pre-training systems have been proposed to fully exploit cluster resources~\cite{um2024metis, yan2024flashflex, sun2024adapipe, li2022amp, zheng2022alpa}. To improve LLM training efficiency, they coordinate model layer allocation across accelerator devices, balance computational workload within each stage, and minimize communication via network-efficient parallelization.
	
	However, these systems still exhibit several limitations that introduce substantial impact on training efficiency:

	\textbf{$\bm{L_1}$: The adopted models cannot characterize training time accurately.} Current heterogeneity-aware LLM training systems~\cite{um2024metis, yan2024flashflex, sun2024adapipe, li2022amp, zheng2022alpa} are commonly built on top of an approximated training time estimation model:
	\vspace{-.5em}
	\begin{equation}
		\label{eq:static_cost_model}
		\sum_{s=1}^{S}t_s + (B-1) \cdot \max_{1\leq j \leq S} t_j + DP_{\text{all}},
		\vspace{-.5em}
	\end{equation}
	where $S$ is the stage count, $B$ is the number of micro-batches, and $DP_{\text{all}}$ represents the DP all-reduce overhead. Each $t_s$ includes the forward, backward, and inter-stage communication costs per micro-batch at stage $s$. However, this formulation results in significant deviation from the ground-truth training time, which adversely impacts training efficiency. To illustrate this issue, we construct an example in Fig.~\ref{fig:deviation_cost_model}, highlighting the discrepancy between the estimated training time in Eq.~\eqref{eq:static_cost_model} and the actual training time. Specifically, in this example, involving two micro-batches pipelined across two stages, the estimation model fails to account for the overlap between communication and computation, substantially overestimating training time. As shown in Fig.~\ref{fig:accurayc_loss_metis}, such models consistently produce inaccurate training time estimates across various cluster configurations, with deviations of up to $1.32\times$. Consequently, these estimation errors can distort the ranking of candidate configurations, misleading the search process and ultimately reducing training efficiency.
	

	\begin{figure}[t]
		\centering
		\includegraphics[width=0.4\textwidth]{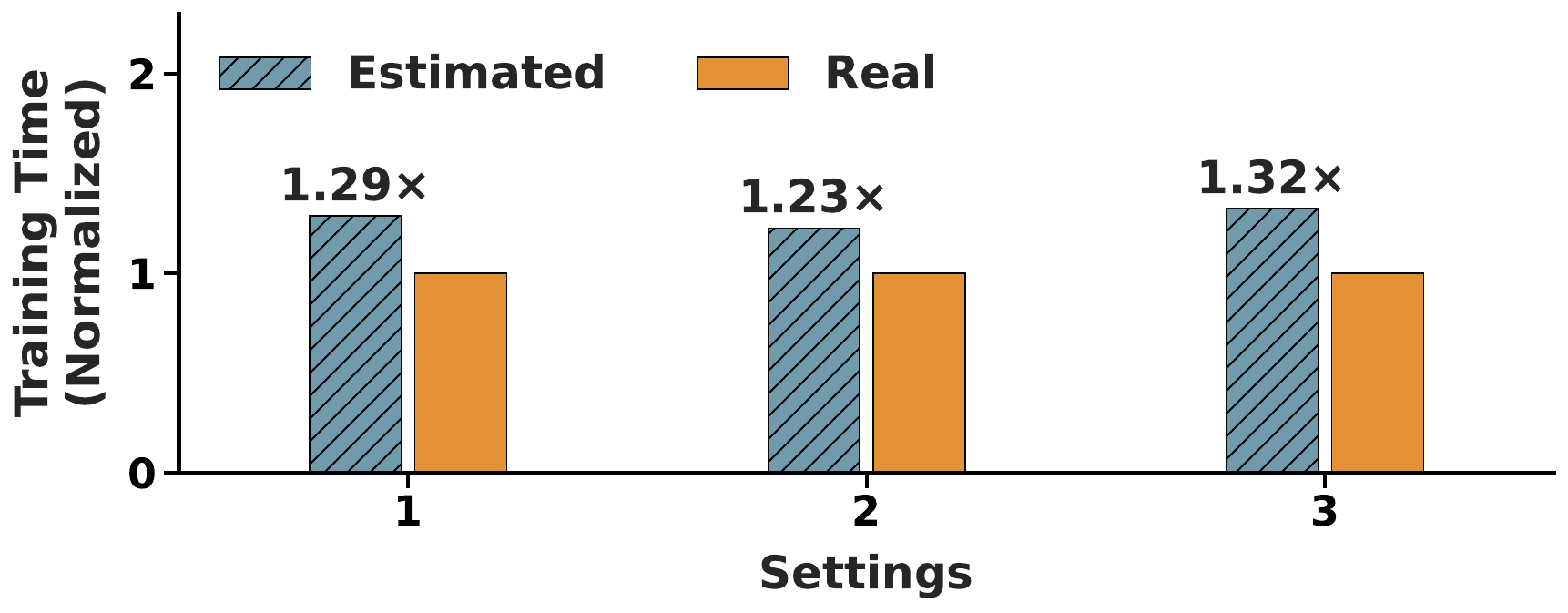}
		\vspace{-1em}
		\caption{The cost model adopted by existing training systems consistently exhibits a significant training time gap.}
		\label{fig:accurayc_loss_metis}
		\vspace{-1.5em}
	\end{figure}
	
	\textbf{$\bm{L_2}$: The adopted pruning strategies yield suboptimal training efficiency.} Exploring vast configuration spaces is computationally expensive, so existing approaches often balance training and search efficiency. However, such compromises risk overlooking the optimal configuration, ultimately degrading training performance. For instance, the pruning strategy in Metis \cite{um2024metis} imposes strong assumptions on the low variance of per-stage device counts,~and on the precedence of data parallelism over tensor parallelism during search. However, the optimal configuration varies with the training hyperparameters, such as batch size and the inter-accelerator connection in clusters,~making these rigid strategies potentially inefficient. Fig.~\ref{fig:heuristic_pruning} shows the notable time gap for training one iteration over LLaMA-3 models \cite{grattafiori2024llama, touvron2023llama} in the cluster consisting of 16 Atlas A2-2 and 16 Atlas A2-3 NPUs~{(detailed NPU parameters are provided in \Cref{sec:experiment_setup})}. The results compare the training time using configurations derived from existing pruning strategies with that of the ideal configuration found via exhaustive search. Specifically, Metis and HexiScale~\cite{yan2024flashflex} exhibited up to $1.91\times$ and $1.47\times$ degradation, respectively, from improper pruning.

	\textbf{$\bm{L_3}$: Existing modeling techniques lack generalizability.} Since existing works employ coarse-grained analytical approximations, their models cannot accommodate diverse pipeline schedules. For example, Eq.~\eqref{eq:static_cost_model} fails to accurately model training time under the Interleaved 1F1B paradigm~\cite{narayanan2021efficient}, which minimizes pipeline bubbles through enhanced computation overlap and thus reduces training time. This same limitation applies to other 1F1B variants like Eager 1F1B~\cite{zhuang2023optimizing} and Seq1F1B~\cite{sun2024seq1f1b}, as their similar bubble-reduction mechanisms likewise elude coarse-grained modeling.
	
	\begin{figure}[t] 
		\centering 
		\begin{subfigure}{0.235\textwidth} 
			\centering
			\includegraphics[width=\textwidth]{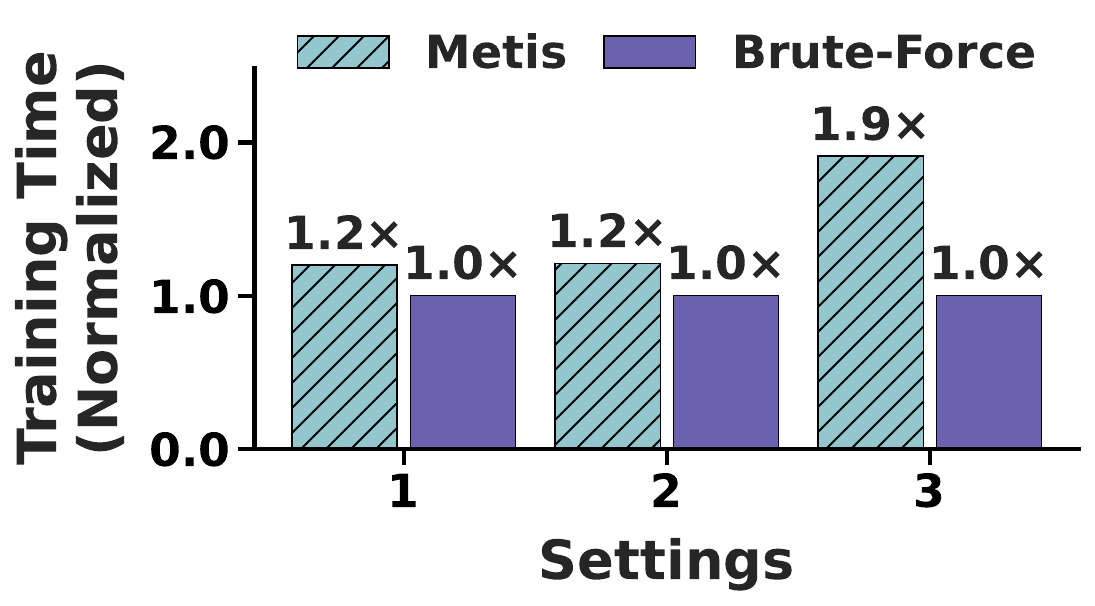}
			\vspace{-1.75em}
			\caption{Degradation in Metis.}
			\label{fig:heuristic_pruning_metis}
		\end{subfigure}
		\hfill
		\begin{subfigure}{0.235\textwidth} 
			\centering
			\includegraphics[width=\textwidth]{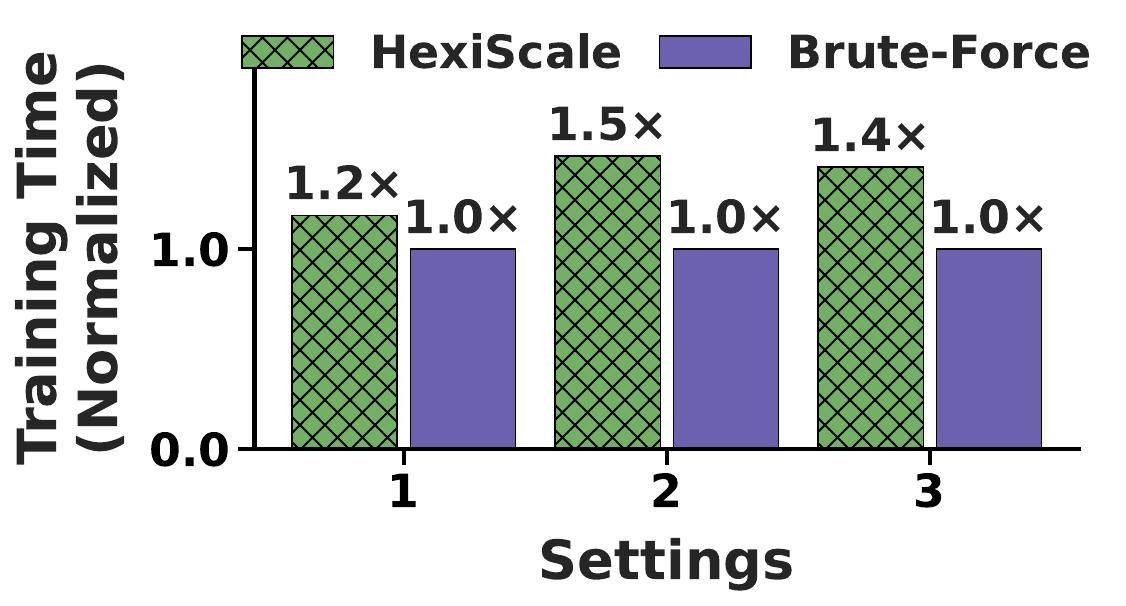}
			\vspace{-1.75em}
			\caption{Degradation in HexiScale.}
			\label{fig:heuristic_pruning_hexiscale}
		\end{subfigure}
		\vspace{-2em}
		\caption{The gap between training time under configurations found by existing systems’ search strategies and under their optimal configurations obtained via brute-force.}
		\label{fig:heuristic_pruning} 
		\vspace{-1em}
	\end{figure}

	\subsection{Graph-based Modeling: Opportunities and Challenges}
	\label{oppor-challenge}
	Based on the preceding analysis of training time evaluation, effective modeling should provide accurate time estimates and capture the characteristics of different training paradigms.
	Fortunately, we found that modeling LLM training as a graph offers a valuable property that significantly enhances modeling performance \cite{moritz2018ray,jia2019beyond,wang2022overlap,won2023astra,tang2024fusionllm,hsia2024mad,jhoo2025pfeife,wan2025coflow,liang2025hapt}. Specifically, across any training paradigm, total LLM training time is determined by the maximum execution time across all micro-batches. Therefore, accurate training time estimation can be achieved by analyzing the computation dependencies across micro-batches and the computation time within each stage. Graph-based modeling techniques are particularly effective in capturing various types of dependencies inherent in LLM training. For instance, computation on the $(k+1)$-th stage can only begin after the activations from the $k$-th stage are delivered, and operations on a given device must wait for preceding computations to complete. By capturing these dependencies, the model can be constructed in a more fine-grained and precise manner, leading to improved accuracy in performance prediction and optimization.
	
	While graph-based modeling is highly effective in capturing complex dependencies, finding an optimal training configuration remains computationally expensive. The primary challenge lies in the vast space of feasible configurations, which renders exhaustive search impractical. To address this, a highly efficient search method is essential—one that can rapidly identify high-quality configurations without traversing numerous candidates, posing a significant challenge in heterogeneous environments.

	\section{Poseidon Overview}
	
	In light of the opportunities and challenges above, we introduce Poseidon, a scalable LLM training system for heterogeneous accelerator resources. Specifically, Poseidon leverages an efficient training time estimation model and search strategy to enhance training efficiency by identifying more effective parallelization configurations while maintaining minimal search overhead. It also generalizes across different pipeline schedules, including 1F1B, Eager 1F1B, and Interleaved 1F1B.

	\subsection{Key Design Ideas}
	In this subsection, we highlight the key ideas underlying Poseidon to explain its performance advantages.
	
	
	$\bm{I_1}:$ \textbf{Explicit DAG construction for training time modeling.} To achieve precise training time modeling, Poseidon utilizes a DAG-based technique to capture the complex dependencies and represent the execution times among operations. Specifically, this approach characterizes the execution at both the stage- and batch-level, allowing us to simultaneously incorporate the data execution order on each stage and the data dependencies within each micro-batch. In this DAG representation, each node corresponds to either an FP or BP operation for a specific micro-batch on a given stage, with the node weight indicating the computation overhead. Edges in different directions represent two types of execution dependencies: inter-stage activation transfers and the execution of different batches within the same stage. The edge weights denote the latency incurred by these dependencies. Consequently, the training time for each iteration can be represented as the longest path through the graph. This approach inherently captures the overlap between communication and computation. 
	
	\begin{figure}[t] 
		\vspace{1em}
		\centering
		\includegraphics[width=0.49\textwidth]{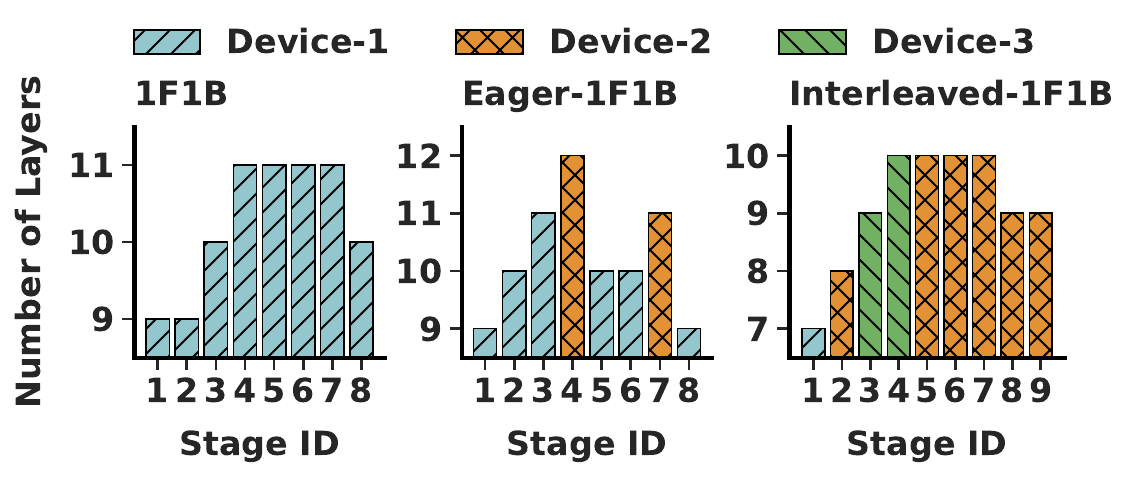}
		\vspace{-2.25em}
		\caption{Optimal parallelization configurations consistently exhibit a ridge-like pattern, i.e., the number of layers allocated to the same device type first increases and then decreases.} 
	\vspace{-1.25em}
	\label{fig:unimodal_example} 
\end{figure}

$\bm{I_2}:$ \textbf{Stage-level pruning via early stopping with partial estimation.} 
To address the challenge of achieving high search efficiency, Poseidon introduces a stage-level pruning strategy to significantly reduce the search overhead. This approach employs an early exit mechanism that halts exploration of inefficient configurations when stages are partially allocated. The mechanism is driven by a lower-bound function we design, which estimates whether the training time of a partially determined configuration is likely to exceed the best result observed so far. If this condition is met, further exploration of that configuration path can be safely skipped, substantially saving search time without compromising solution quality.


$\bm{I_3}:$ \textbf{Layer-level pruning via leveraging ridge-like distribution.} 
To further reduce the search overhead in identifying optimal configurations, Poseidon incorporates an efficient layer-level pruning approach that maps model layers across stages. This leverages a key observation: in the optimal parallelization configuration, the number of layers assigned to all stages of the same device type consistently follows a ridge-like pattern, where the count first increases and then decreases, as shown in Fig.~\ref{fig:unimodal_example}. Exploiting this property, Poseidon can skip the evaluation of certain configurations without compromising solution quality. For instance, once the number of layers allocated to a particular device type starts to decrease, configurations with a higher layer count than the previous stage can be safely pruned, substantially reducing search time. 

\subsection{System Architecture}
The Poseidon system comprises three functional modules: \emph{profiling}, \emph{searching}, and \emph{simulator}. The representation of these modules and their interaction are depicted in Fig.~\ref{Fig.overview}.

Poseidon first incorporates a profiling module in the system to conduct computation of model layers and collect the corresponding execution time information. With the help of this module, we can obtain the computation time, communication bandwidth, and memory footprint—including activation and gradient storage—for a single layer under different TP degrees across heterogeneous accelerator devices. Additionally, the memory profiling provides critical constraints on the maximum number of layers that can be assigned to each stage. This profiling data underpins the DAG-based cost model, enabling accurate evaluation of candidate parallelization configurations.

Given the vast search space introduced by accelerator heterogeneity, Poseidon’s search module explores optimal training parallelization configurations by applying two complementary pruning strategies. It first uses stage-level pruning to early-terminate branches whose projected training time is clearly suboptimal, and then exploits the insight that, within each accelerator type, optimal layer allocations across stages follow a ridge-like pattern to perform layer-level pruning. Together, these techniques swiftly eliminate low-efficiency candidates and dramatically accelerate the search process.

Moreover, the search process relies on feedback from a simulator module, which uses a DAG to model the computation and communication dependencies of each parallelization configuration. By applying a dynamic programming technique over the DAG, the simulator module consistently captures training time information from the path with the longest time overhead for the provided parallelization configurations. 

\begin{figure}[t] 
	\vspace{0.5em}
	\centering 
	\includegraphics[width=0.49\textwidth]{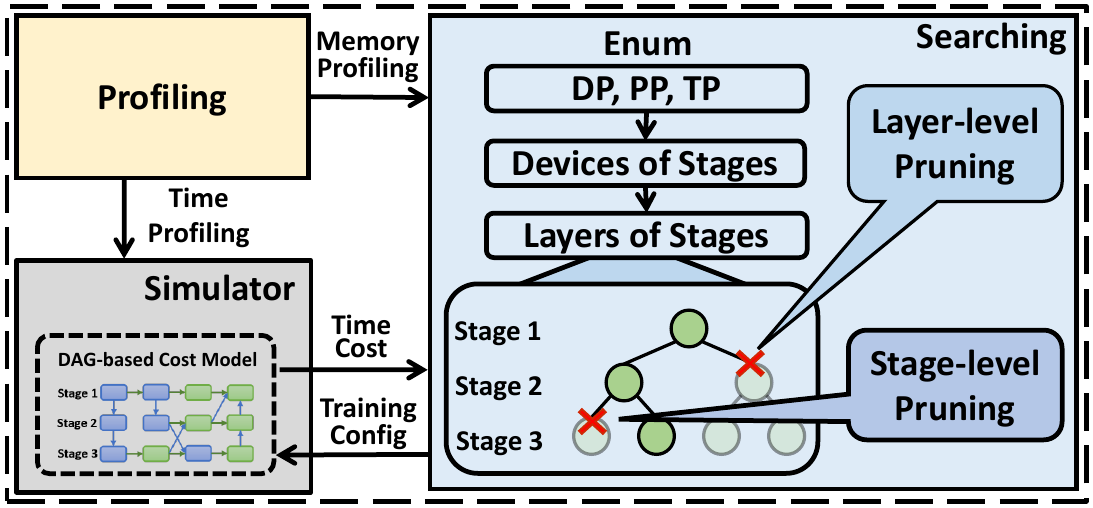}
	\vspace{-1.25em}
	\caption{The system overview of Poseidon.}
	\vspace{-1.5em}
	\label{Fig.overview}
\end{figure}

\begin{figure*}[t] 
	\vspace{-1em}
	\centering 
	\includegraphics[width=0.9\textwidth]{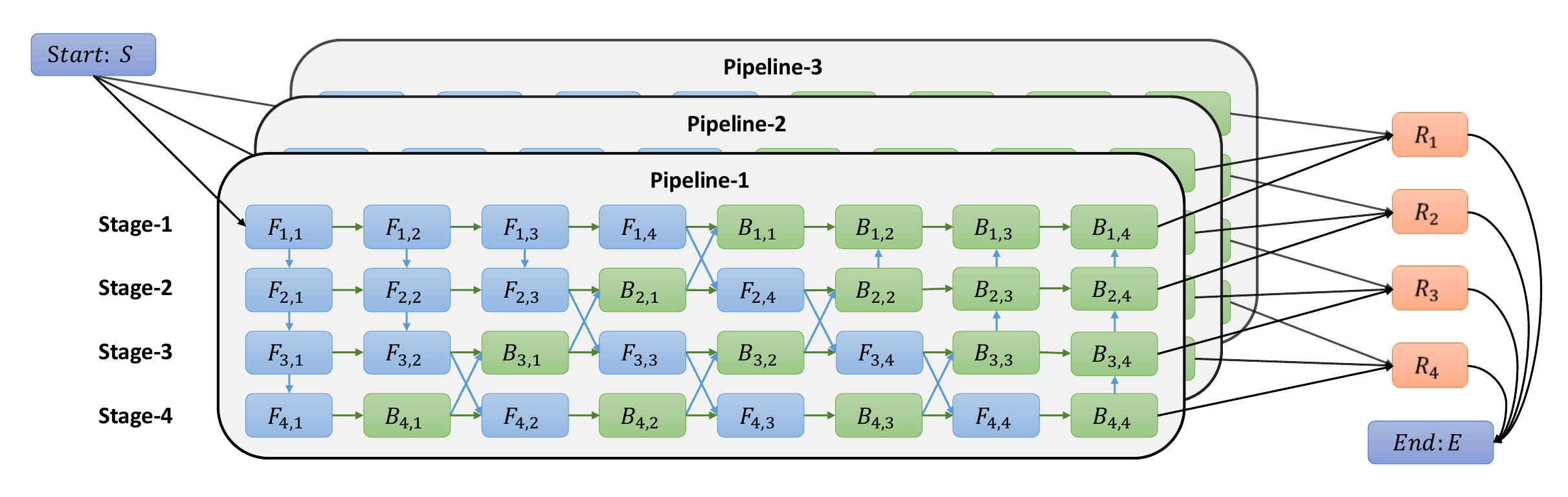}
	\vspace{-1.5em}
	\caption{Graph-based pipeline cost model (3 pipelines, 4 stages per pipeline, 4 micro-batches per pipeline, 1F1B schedule) with computational graph annotations: $F_{s,m}$/$B_{s,m}$ denote forward/backward of micro-batch $m$ at stage $s$, $R_s$ indicates DP all-reduce; blue edges represent inter-stage dependencies and green edges represent intra-stage dependencies.}
	\vspace{-1em}
	\label{Fig.PEG}
\end{figure*}

\section{Design Details of Poseidon}

This section presents Poseidon’s design, starting with a discussion of the DAG-based modeling technique in~\Cref{graph_based_cost_model}. We then introduce the stage-level and layer-level pruning strategies in \Cref{stage_level_pruning} and \Cref{layer_level_pruning}. Finally, we propose an efficient search algorithm in~\Cref{sec:search_policy}, which leverages both strategies to accelerate the search process without compromising training efficiency.


\subsection{DAG-based Cost Model Design}
\label{graph_based_cost_model}

Before presenting our DAG-based training time model, we first contextualize it within prior work. DAG abstractions have been widely used to model DNN and LLM training pipelines for purposes such as compilation optimization \cite{moritz2018ray,wang2022overlap}, execution scheduling \cite{tang2024fusionllm,hsia2024mad,jia2019beyond,jhoo2025pfeife}, communication optimization \cite{wan2025coflow,wang2022overlap,won2023astra}, and formal reasoning \cite{liang2025hapt}. These works convincingly demonstrate the expressive power of DAGs. 

In Poseidon, we leverage this expressive power for a fundamentally different purpose. Our contribution is not the DAG abstraction itself, but its novel application as a quantitative, microbatch-level cost model for guiding optimal parallelism search in LLM training. While prior work uses DAGs to represent the system for the purpose of optimization or reasoning, we use the DAG to simulate and estimate training time. 

We formulate our training time modeling technique as follows: given a parallelization configuration $C$ with $D$ pipelines, $N$ stages per pipeline, and $M$ micro-batches, the model yields a precise training time estimation $T$. The modeling process involves three steps: node construction, edge construction, and weight assignment. Once $C$ is represented by this DAG, we compute $T$—defined as the critical path length—using dynamic programming. A representative example illustrating this DAG-based cost model is provided in Fig.~\ref{Fig.PEG}.

\subsubsection{Node construction} In the DAG-based cost model, nodes represent diverse operations. In particular, we model the training process with the following node types based on event type:

\textbf{Micro-batch computation nodes.} For each micro-batch at every stage, two nodes are generated: a \textit{forward pass node} \( F_{s,m} \) and a \textit{backward pass node} \( B_{s,m} \), where \( s \in [1,N] \) represents the stage index and $m$ is the micro-batch index. This results in \( 2NM \)  computation nodes for each pipeline.

\textbf{DP all-reduce nodes.} Each stage includes a synchronization node \( R_s \) for gradient aggregation, resulting in a total of \( N \) communication nodes.

\textbf{Control nodes.} A virtual \textit{start node} \( S \) and \textit{end node} \( E \) define the global execution boundaries.


\subsubsection{Edge construction} 
We leverage directed edges to capture dependencies among nodes in the DAG. Specifically, we define the dependencies with the following types of edges:

\textbf{Inter-stage edges.} For each micro-batch \( m \), the forward computation flow is represented by  \( F_{s,m} \rightarrow F_{s+1,m} \), while the backward propagation flow is represented by \( B_{s,m} \rightarrow B_{s-1,m} \). This requires \( 2(N-1)M \) edges for each pipeline.

\textbf{Intra-stage edges.} The pipeline schedule, such as 1F1B, dictates micro-batch execution order within a stage. For \( M \) micro-batches, each stage needs \( 2M-1 \) edges to represent this order, totaling \( N(2M-1) \) edges per pipeline.

\textbf{All-reduce edges.} Each \( R_s \) depends on \( D \) pipeline instances completing their final backward pass \( B_{s,M} \). This requires \( ND \) edges, where the edges are represented as \( B_{s,M} \rightarrow R_s \).

\textbf{Auxiliary edges.} The start node \( S \) connects to all first-stage forward nodes \(F_{1,m}\) via \( D \) edges, and all \( R_s \) nodes converge to the end node \( E \) via \( N \) edges, yielding \( D+N \) edges.

\subsubsection{Weight assignment according to profiling data} 
Poseidon assigns weights to nodes and edges in the graph based on detailed profiling data. Most LLMs consist of repeatedly stacked homogeneous layers, such as LLaMA~\cite{touvron2023llama} and Qwen3-32B~\cite{yang2025qwen3technicalreport}, so profiling a single representative layer is sufficient. Some models adopt mixed layer designs; for example, DeepSeek-V3~\cite{deepseekai2025deepseekv3technicalreport} contains both dense Transformer layers and MoE Transformer layers. In such cases, Poseidon profiles one representative layer for each distinct layer type, which introduces only marginal overhead due to the small number of distinct layer types. To support the cost model, Poseidon profiles three types of performance metrics under different accelerator types and TP degrees: (1) the forward and backward execution times of a single layer, (2) the activation transfer time, and (3) the DP all-reduce synchronization time of a single layer. These per-layer profiling results are then aggregated to estimate the stage-level execution and communication costs for weight assignment as follows:

\textbf{Computation nodes \( \bm{F_{s,m}}, \bm{B_{s,m}} \).} The forward and backward pass times are assigned to these nodes. Each value equals the profiled per-layer computation time multiplied by the number of layers in stage $s$. 

\textbf{Inter-stage edges.} The weights of these edges represent activation transfer time, obtained from profiling under the corresponding device type and TP degree.

\textbf{Intra-stage edges.} These edges have zero weight and represent only the batch execution order within a stage.

\textbf{DP all-reduce nodes \( \bm{R_s} \).} Gradient synchronization time is assigned to these nodes, which equals the per-layer all-reduce time multiplied by the number of layers in stage $s$.

\textbf{Control nodes and auxiliary edges.} These are assigned zero weight, indicating no time cost.





\subsubsection{Modeling various pipeline schedules} 
Pipeline schedules such as Eager 1F1B and Interleaved 1F1B have different execution dependency structures. The flexible DAG-based model accommodates these variations by adjusting node counts and dependency edges without changing the overall framework. We use Interleaved 1F1B to illustrate the corresponding modifications to node and edge construction.


\textbf{Node construction.} In the Interleaved 1F1B schedule with a virtual pipeline size of $V$, there are $VN$ virtual stages. Thus, the number of computation nodes increases to $2VNM$.

\textbf{Edge construction.} As the number of nodes increases under the Interleaved 1F1B schedule, the number of inter-stage edges also scales accordingly, reaching $2(NV-1)M$, and the number of intra-stage edges increases to $N(2VM-1)$.





\subsection{Search Policy for Optimal Parallel Training}
\label{search_algorithm}
Due to the expansive search space from heterogeneous accelerator clusters, effective pruning techniques are essential to identify the most efficient training configuration.


\subsubsection{Stage-level pruning}
\label{stage_level_pruning}
Instead of applying heuristics to trade off between training efficiency and search overhead, we first propose a stage-level pruning strategy to maintain an optimally efficient configuration with reduced search time. Specifically, it enables early exit in branches deemed to be inefficient, even if only partial stages of them are determined, to avoid unnecessary configuration exploration. This capability is critical as today's LLMs consist of dozens of layers, which inevitably expands the search space across stages.

To achieve this, Poseidon introduces a lower-bound function $g(S_x)$, used to check if the current search is worth continuing. Here, $S_x$ denotes the temporal parallelization configuration with only the number of model layers in stages $1$ to $x$ determined. The lower-bound function is defined as:
\vspace{-.5em}
\begin{equation}
	\label{eq:stage-level-pruning}
	g(S_x) = \sum _{i=1}^{x-1}f_i + M\cdot \left(f_x+b_x\right) + \max\limits_{1\leq i \leq x} \left\{\sum_{j = i}^{x-1} b_j+d_i \right\}, 
	\vspace{-.5em}
\end{equation}
where $f_i$, $b_i$, and $d_i$ denote the forward, backward, and DP all‐reduce time for stage $i$, and $M$ is the number of micro‐batches. This formulation provides an explicit lower bound of training time given the determined layer allocation under $S_x$, where adjusting the allocation of remaining layers can only yield a longer training time. If the temporal configuration $S_x$ yields a training time longer than $T_{\text{min}}$, i.e., the best training time observed so far, Poseidon stops further exploration based on this configuration. 
Otherwise, Poseidon continues exploring configurations extended from $S_x$ and updates $T_{\text{min}}$ only when a complete configuration with a shorter estimated training time is found.
Due to the exponential growth of the search space with respect to the number of stages traversed, the early stopping strategy effectively mitigates the computational overhead of exhaustive parallel training search. We formalize this observation in the subsequent theorem, with a full proof by contradiction detailed in the Appendix \ref{sec:poseidon_stage_pruning_guarantee}.

\begin{theorem}
	\label{theorem.optimality-based pruning}
	If $g(S_x) > T_{\min}$,  all extensions of $S_x$ can be safely pruned without loss of optimality.
\end{theorem}



\subsubsection{Layer-level pruning} 
\label{layer_level_pruning}
In addition to stage-level pruning, Poseidon further introduces layer-level pruning to improve search efficiency. These two strategies operate at different levels and can be effectively integrated with each other. As shown in Fig.~\ref{fig:unimodal_example}, the optimal parallelization configuration for LLM training consistently exhibits a ridge-like layer allocation across devices of the same type, i.e., the number of layers assigned to stages within a specific accelerator type first increases and then decreases. This ridge-like pattern reveals a key property guiding our pruning strategy. However, one may wonder if this property will always hold, especially as the search space grows exponentially with cluster size and number of accelerator types. Therefore, a theoretical guarantee is essential for search efficiency based on this property. Accordingly, we propose a rigorous theorem that works for any heterogeneous environment at an arbitrary scale as follows:

\begin{theorem}
	\label{theorem.Unimodal}
	Given $N$ pipeline stages, there exists an optimal layer allocation scheme $l_1, l_2, \dots, l_N$, where $l_i$ denotes the number of layers assigned to stage $i$, such that for any set of device groups of the same type assigned to $k$ stages $s_1 < s_2 < \cdots < s_k$, the corresponding subsequence $l_{s_1}, l_{s_2}, \dots, l_{s_k}$ forms a ridge-like distribution.
\end{theorem}

This theoretical result provides a principled foundation for search and enables the pruning of suboptimal configurations. An intuitive explanation of Theorem~\ref{theorem.Unimodal} is as follows: earlier stages need to cache more forward activations~\cite{narayanan2019pipedream}, resulting in higher memory pressure on devices of the same type. To alleviate this, fewer layers are allocated to early stages, leading to gradually more layers toward intermediate stages. Conversely, latter stages wait longer for their first micro-batch; assigning more layers here prolongs training time and creates a bottleneck. Thus, the number of layers should decrease toward the tail of the model. This dual pressure—balancing memory constraints in early stages and latency bottlenecks in later ones—creates a ridge-like layer allocation across stages of the same device type. We defer the full proof to the Appendix. 

Furthermore, this theorem reveals a counterintuitive finding: merely balancing computation time across stages increases training time. Consequently, existing approaches focused on minimizing stage time variance are inherently less efficient. 

\subsubsection{Search policy}
\label{sec:search_policy}

Building upon the properties described above, Poseidon performs an efficient search over parallelization configurations. Since TP requires frequent communication, Poseidon restricts each TP group to a single node to avoid inter-node bandwidth bottlenecks~\cite{um2024metis, sailor, shoeybi2019megatron}, whereas DP and PP are not subject to this restriction. To mitigate straggler effects, where faster pipelines are blocked by peers during gradient aggregation, Poseidon enforces identical configurations across all pipelines.

Before the main search, Poseidon randomly samples 500 feasible configurations and estimates their training times using the DAG-based cost model. The best sampled result initializes $T_{\text{min}}$, providing an effective upper bound for subsequent pruning. The impact of the number of warm-up samples is evaluated in \Cref{sec:impact-of-warm-up-sampling-count}.

Poseidon then enumerates feasible 3D parallelism configurations and device-to-stage assignments while recursively determining the layer allocation of each stage. For every complete configuration, the simulator estimates its training time and updates $T_{\text{min}}$ when a better configuration is found. During this process, Poseidon applies stage- and layer-level pruning cooperatively.
At stage $x$, if the lower bound $g(S_x)$ of a partial configuration exceeds $T_{\text{min}}$, the entire search branch is discarded. Meanwhile, the layer-level strategy explores only allocations satisfying both memory constraints and the theoretically derived ridge-like pattern. The two strategies are complementary: stage-level pruning is particularly effective for configurations with many pipeline stages, whereas layer-level pruning substantially reduces the search space when many candidate layer allocations exist for each stage.

\section{Poseidon Implementation}
The Poseidon framework consists of three core components: a search module, a simulator module, and a profiling module. To meet stringent performance requirements during the search process, the search and simulator modules are implemented in C++ over 3,000 lines of code. The profiling module is implemented in Python and also surpasses 3,000 lines.


\textbf{Simulator design.} The simulator module provides precise training time estimates for a given parallelization configuration, with a user-specified batch size and sequence length. It evaluates the training time using profiling data and the DAG-based cost model. Specifically, if the PP degree remains unchanged, the DAG does not need to be rebuilt; only the edge and node weights are updated.

\textbf{Training on NPUs}. To fully unleash the potential of NPUs, Poseidon performs profiling and executes LLM training on top of MindSpeed-LLM~\cite{ascend2025mindspeed}, a high-performance training framework built on Megatron-LM~\cite{narayanan2021efficient} and optimized for NPUs. Consequently, Poseidon provides full-stack NPU support for LLM training.

\textbf{Automatic training workflow.} Poseidon employs a user-transparent integration mechanism to interface with LLM training frameworks. It injects search and profiling logic into user training requests transparently. The optimal configuration is applied automatically, enabling end-to-end training. 






\section{Evaluation}
We evaluate Poseidon on heterogeneous clusters across five aspects: training throughput, search efficiency, cost model accuracy, pruning effectiveness, and system scalability. In this section, training throughput is defined as the reciprocal of per-iteration training time; therefore, higher throughput indicates shorter training time.

\vspace{-0.5em}
\subsection{Experiment Setup}
\label{sec:experiment_setup}

\textbf{Cluster environments.} Experiments are conducted on both NPU clusters and GPU clusters:\\
(1)~\textit{NPU clusters.} Three NPU clusters are used, equipped with Atlas A2-2, A2-3, and A2-4 devices, respectively. Each cluster consists of homogeneous nodes with eight NPUs per node. The intra-node, intra-cluster inter-node, and cross-cluster bandwidths are 1400~Gbps, 200~Gbps, and 75~Gbps, respectively. The peak performance of the A2-2, A2-3, and A2-4 is 376~TFLOPs, 313~TFLOPs, and 280~TFLOPs, respectively. The A2-2 and A2-3 provide 64~GB of memory, while the A2-4 provides 32~GB.\\
(2)~\textit{GPU clusters.} Three GPU clusters are used, equipped with A100, RTX-4090, and RTX-3090 GPUs, respectively. Each cluster consists of homogeneous nodes with eight GPUs per node. The intra-node, intra-cluster inter-node, and cross-cluster bandwidths are 256~Gbps, 100~Gbps, and 75~Gbps, respectively. The peak performance of the A100, RTX-4090, and RTX-3090 is 312~TFLOPs, 330~TFLOPs, and 71~TFLOPs, with memory capacities of 80~GB, 24~GB, and 24~GB, respectively.


\textbf{Baselines.} We compare Poseidon with state-of-the-art\linebreak heterogeneity-aware pre-training systems, including HexiScale \cite{yan2024flashflex} and Metis \cite{um2024metis}.

\textbf{Experimental settings.} We evaluate Poseidon under a diverse set of models and heterogeneous configurations:
\begin{itemize}[left=0pt, itemsep=1pt, labelsep=0.0em, label=--]
	\vspace{-0.25em}
	\item Setting 1: LLaMA-3(8B)~\cite{grattafiori2024llama} on one 8$\times$Atlas A2-2 node and one 8$\times$Atlas A2-3 node.
	\vspace{-0.25em}
	\item Setting 2: LLaMA-2(13B)~\cite{touvron2023llama2} on two 8$\times$Atlas A2-2 nodes and two 8$\times$Atlas A2-3 nodes.
	\vspace{-0.25em}
	\item $\text{Setting 3}$: LLaMA(30B)~\cite{touvron2023llama} on two 8$\times$Atlas A2-2, two 8$\times$Atlas A2-3, and two 8$\times$Atlas A2-4 nodes.
	\vspace{-0.25em}
	\item Setting 4: Mixtral(8$\times$7B)~\cite{jiang2024mixtral}, an MoE model, on one 8$\times$A2-2 node, two 8$\times$A2-3 nodes, and four 8$\times$A2-4 nodes.
	\vspace{-0.25em}
	\item $\text{Setting 5}$: LLaMA-2(13B) on one 8$\times$A100 node, two 8$\times$RTX-4090 nodes, and two 8$\times$RTX-3090 nodes.
	\vspace{-0.25em}
	\item Setting 6: LLaMA-3(70B) on six 8$\times$Atlas A2-2, six 8$\times$Atlas A2-3, and twenty-four 8$\times$Atlas A2-4 nodes. 
\end{itemize}

\begin{figure}[tbp] 
	\vspace{-0.5em}
	\centering 
	\includegraphics[width=0.49\textwidth]{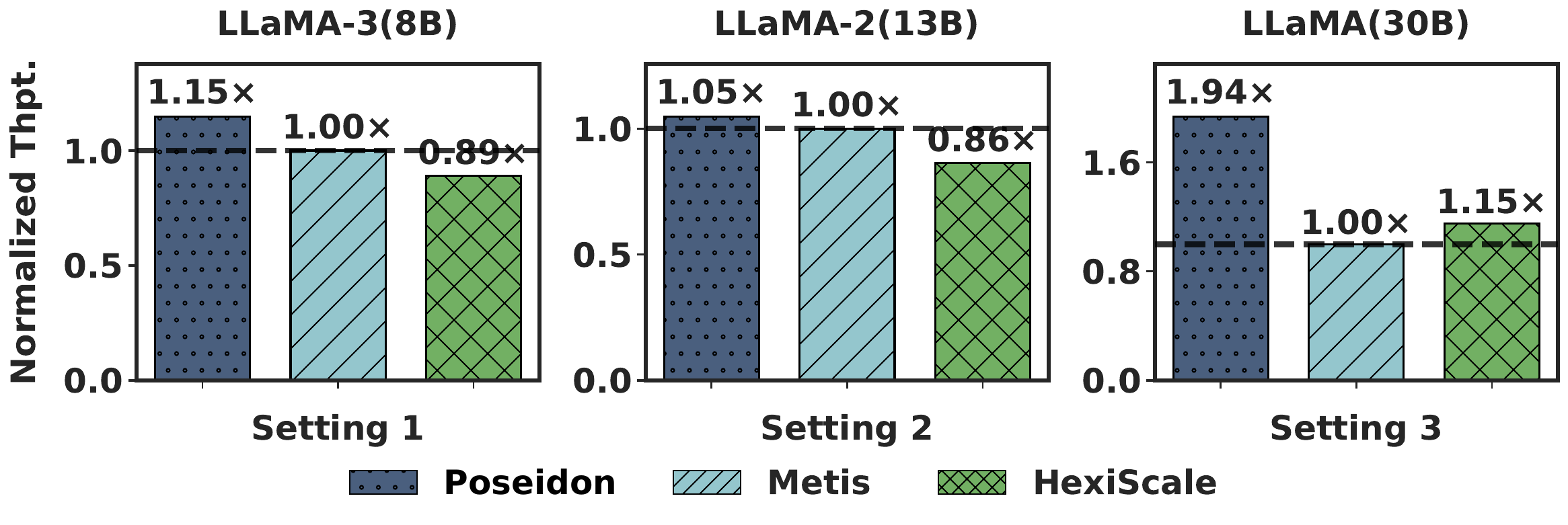}
	\vspace{-2.5em}
	\caption{The training throughput of Poseidon compared to baseline systems, across different models and cluster settings.}
	\label{Fig.experiment-comparison-throuput}
	\vspace{-.5em}
\end{figure}

\subsection{End-to-end Training Performance}
\label{sec:end-to-end-performance}

\subsubsection{Throughput across model and cluster scales}  
For a fair comparison, all systems adopt the 1F1B schedule, the only one supported by the baselines. As shown in Fig.~\ref{Fig.experiment-comparison-throuput}, Poseidon achieves the highest training throughput across all settings, outperforming Metis by 1.05–1.94$\times$ and HexiScale by 1.21–1.69$\times$, with both maximum gains observed on LLaMA-30B. Metis's layer allocation suffers from two issues: (1) heuristics that prioritize computational balance across stages, and (2) iterative error-correction for OOM failures. In 1F1B scheduling, its greedy allocation often triggers OOMs because earlier stages require higher activation memory. As a result, the subsequent heuristic reallocation lacks global optimality and degrades the quality of final allocations. Similarly, HexiScale relies on a two-phase graph-partitioning heuristic without optimality guarantees. In contrast, Poseidon exhaustively explores the search space under memory constraints, leveraging an accurate DAG-based cost model and provable pruning strategies to consistently find optimal configurations.

\subsubsection{Analysis of search overhead}
All methods require a one-time profiling step to collect data for the cost model. In our evaluation, this step takes up to 2 hours and is reused across all methods. Thus, the profiling time is excluded from the reported search overhead.

The brute-force search over the full configuration space is prohibitively time-consuming. For example, under Setting~1, there are approximately $1.01 \times 10^9$ feasible configurations according to dynamic programming analysis, leading to at least 7 hours of search time given a per-search cost of 0.025~ms on our testbed. The estimated search time further increases to 4.81 days in Setting~2, with each search taking 0.0408~ms. When scaled to Setting~3, the estimated total search time exceeds 5,000 years—each search taking 0.189~ms. Although searching a single configuration is on the order of milliseconds, the enormous size of the configuration space renders brute-force search impractical. These results highlight the need for efficient search algorithms that reduce exploration time while preserving solution quality, as discussed in \Cref{oppor-challenge}.

Fig.~\ref{Fig.experiment-comparison-searching} shows the search time taken by each method for the experiments in Fig.~\ref{Fig.experiment-comparison-throuput}. Although Poseidon does not achieve the shortest search time in all settings, it effectively balances search efficiency and configuration quality. For smaller models (LLaMA-8B/13B), differences in search time across methods are marginal, with search times remaining within the same order of magnitude. Considering the overall training cost, this overhead is negligible. For instance, assuming a training dataset of 10 billion tokens under Setting 1, the configuration selected by Poseidon completes training in approximately 21.2 hours, which is 2.8 hours faster than that produced by Metis. Such throughput gains become increasingly significant as the size of the datasets scales up. Moreover, since brute-force search is feasible under Settings~1 and~2, we employ brute-force search in these cases. The results show that Poseidon identifies exactly the same optimal parallelism configurations as brute-force search, achieving identical throughput.

On larger models (LLaMA-30B), Poseidon reduces search time by up to 11.34$\times$ compared to Metis, whose depth-first search becomes less efficient as the search space grows. Conversely, Poseidon's pruning mechanism aggressively eliminates suboptimal configurations early, reducing search overhead. While HexiScale also achieves rapid search via a two-phase graph-partitioning heuristic, it compromises configuration quality, as shown in Fig.~\ref{Fig.experiment-comparison-throuput}.

\begin{figure}[t] 
	\vspace{-0.75em}
	\centering 
	\includegraphics[width=0.48\textwidth]{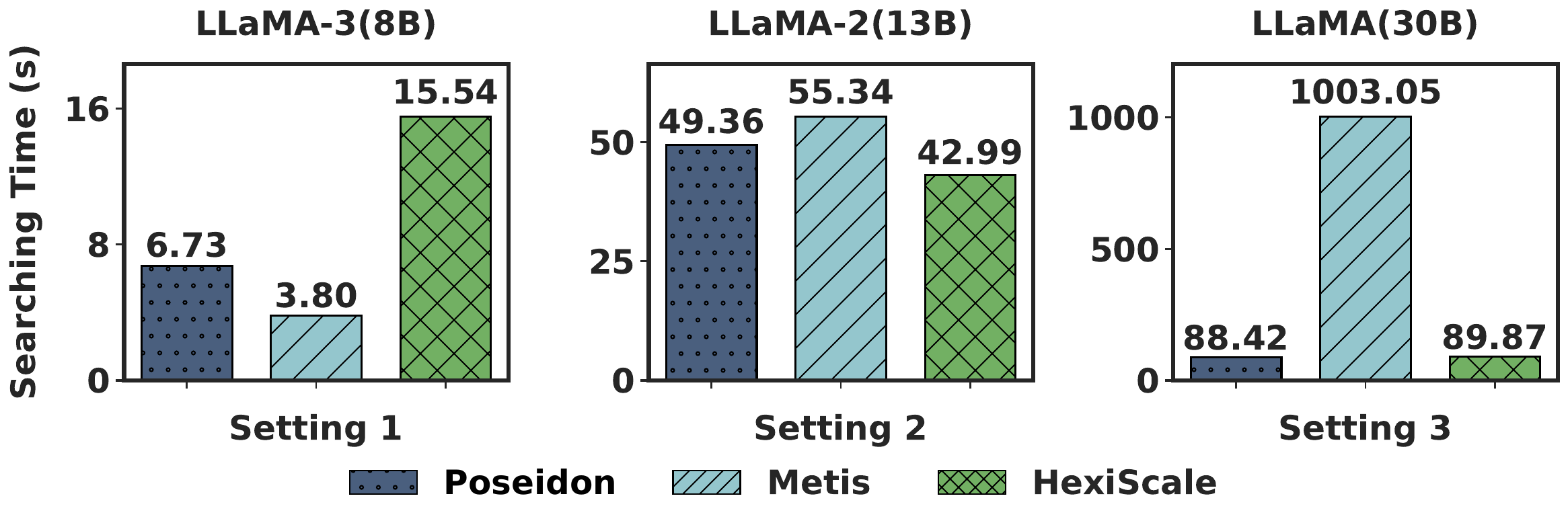}
	\vspace{-2.5em}
	\caption{The search time of all systems across various model scales and heterogeneous cluster settings.}
	\label{Fig.experiment-comparison-searching}
	\vspace{-2em}
\end{figure}

		
		
\begin{table*}[t]
	\centering
	\begin{threeparttable}
		\captionsetup{skip=2pt}
		\caption{Configurations discovered by Poseidon, Metis, and HexiScale}
		\label{tab:config_comparison}
		\renewcommand{\arraystretch}{0.9}  
		\setlength{\tabcolsep}{8pt} 
		\small 
		\begin{tabular}{@{}c c c c@{}}
			\toprule
			\textbf{Methods} & \textbf{Poseidon} & \textbf{Metis} & \textbf{HexiScale} \\
			\midrule
			\text{PP} & 16 & 14 & 10 \\
			\multicolumn{1}{@{}c}{\text{(DP,~TP,~NPU,~Layer)}} & 
			
			\makecell[l]{
				(2,~1,~A3,~2)$\times$1, (2,~1,~A3,~3)$\times$2, \\
				(2,~1,~A3,~2)$\times$1, (2,~1,~A2,~3)$\times$4, \\
				(2,~1,~A3,~2)$\times$2, (2,~1,~A3,~1)$\times$2, \\
				(2,~1,~A2,~3)$\times$4
			} & 
			\makecell[l]{
				(2,~1,~A2,~3)$\times$8, 
				(2,~2,~A3,~3)$\times$2,\\ (2,~1,~A3,~3)$\times$1, 
				(2,~1,~A3,~2)$\times$1,\\ (2,~1,~A3,~3)$\times$1,
				(2,~1,~A3,~2)$\times$1
			} & 
			\makecell[l]{
				(1,~4,~A3,~5)$\times$3, (1,~2,~A3,~2)$\times$2, \\
				(1,~4,~A2,~5)$\times$1, (1,~4,~A2,~6)$\times$2, \\
				(1,~2,~A2,~2)$\times$2
			} \\
			
			\bottomrule
		\end{tabular}
		\begin{tablenotes}[flushleft]
			\footnotesize
			\item[] \textit{Notation:} $(a,b,c,d)\times n$ denotes $n$ consecutive stages with identical configurations.
		\end{tablenotes}
	\end{threeparttable}
	\vspace{-0.5em}
\end{table*}

\subsubsection{Case study of configuration results}
Table \ref{tab:config_comparison} presents the parallelism strategies discovered by all methods under Setting 2. Poseidon identifies a ridge-like configuration through fine-grained device and layer allocation. Metis produces similar DP and TP settings but retains a coarse pipeline layout due to its limited search space, which hinders exploitation of inter-device heterogeneity. HexiScale limits configurability by grouping similar devices using a graph-partitioning heuristic, narrowing opportunities for better configurations.

\subsection{Evaluation of the DAG-Based Cost Model}
\label{sec:accuracy-of-dag-based-cost-model}

\subsubsection{Training time estimation in end-to-end experiments}
We compare Poseidon’s DAG-based estimation with actual end-to-end training times. For actual iteration time, we run 30 iterations, discard the first five to avoid warm-up, and average the remaining 25. As shown in Fig.~\ref{Fig.DAGacc}, Poseidon maintains low estimation errors ($\leq$2\%) across Settings 1–3, indicating DAG-based cost model accurately captures runtime behaviors, including computation–communication overlap and dependencies across heterogeneous devices. 

We further evaluate its fidelity on LLaMA-3 (70B) under 20 distinct heterogeneous deployments, using two nodes of 8×Atlas A2-2, four nodes of 8×Atlas A2-3, and five nodes of 8×Atlas A2-4. For each deployment, 300 configurations are sampled (100 per pipeline schedule: 1F1B, Interleaved 1F1B, Eager 1F1B). These configurations cover diverse DP/PP/TP degrees (DP=1–4, PP=1–16, TP=1–8) to reflect realistic intra- and inter-node parallelism. Across all configurations, the simulator achieves an average prediction accuracy of 98\%.

\subsubsection{Generalizability to pipeline schedules}
We evaluate Poseidon under different pipeline schedules using Setting 2. As shown in Fig.~\ref{Fig.combined_throughput}, varying the pipeline schedule alone can change throughput by up to 1.3× with other factors held constant. This highlights the importance of pipeline schedule as an additional optimization dimension—Poseidon fully exploits this flexibility, whereas existing baselines cannot.

\begin{figure}[t]
	\vspace{-0.75em}
	\centering
	\begin{minipage}{0.24\textwidth}
		\centering
		\includegraphics[width=\textwidth]{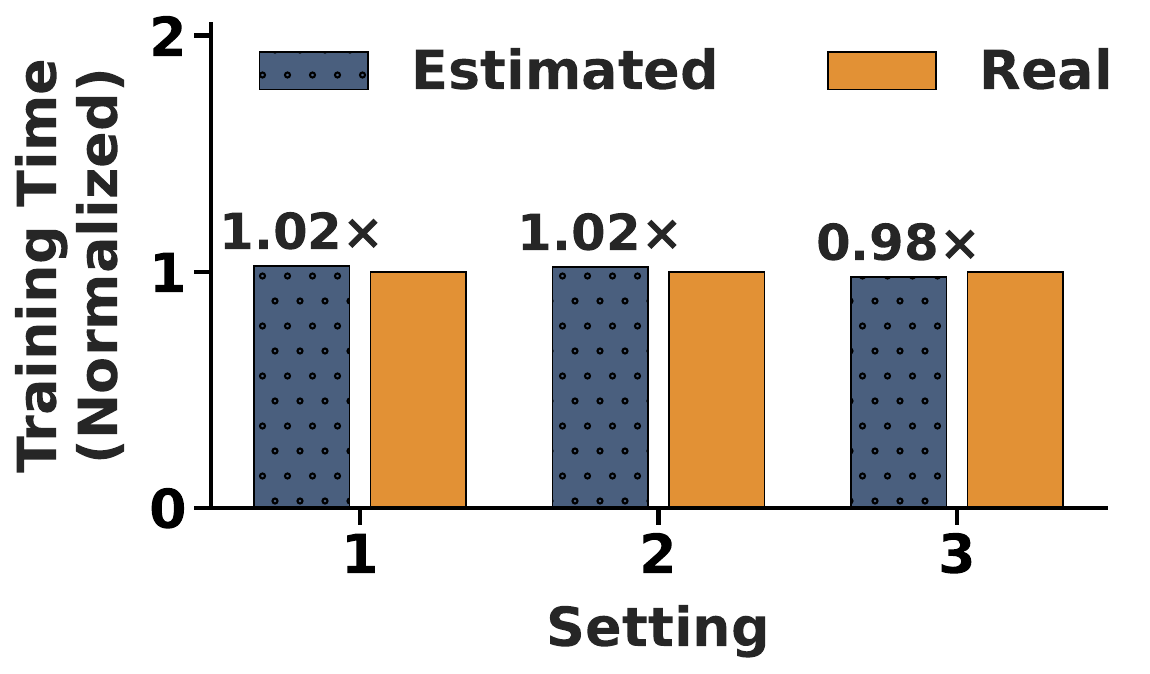}
		\vspace{-2em}
		\caption{Accuracy of the DAG-based model in estimating training time.}
		\label{Fig.DAGacc}
	\end{minipage}
	\hfill
	\begin{minipage}{0.23\textwidth}
		\centering
		\includegraphics[width=\textwidth]{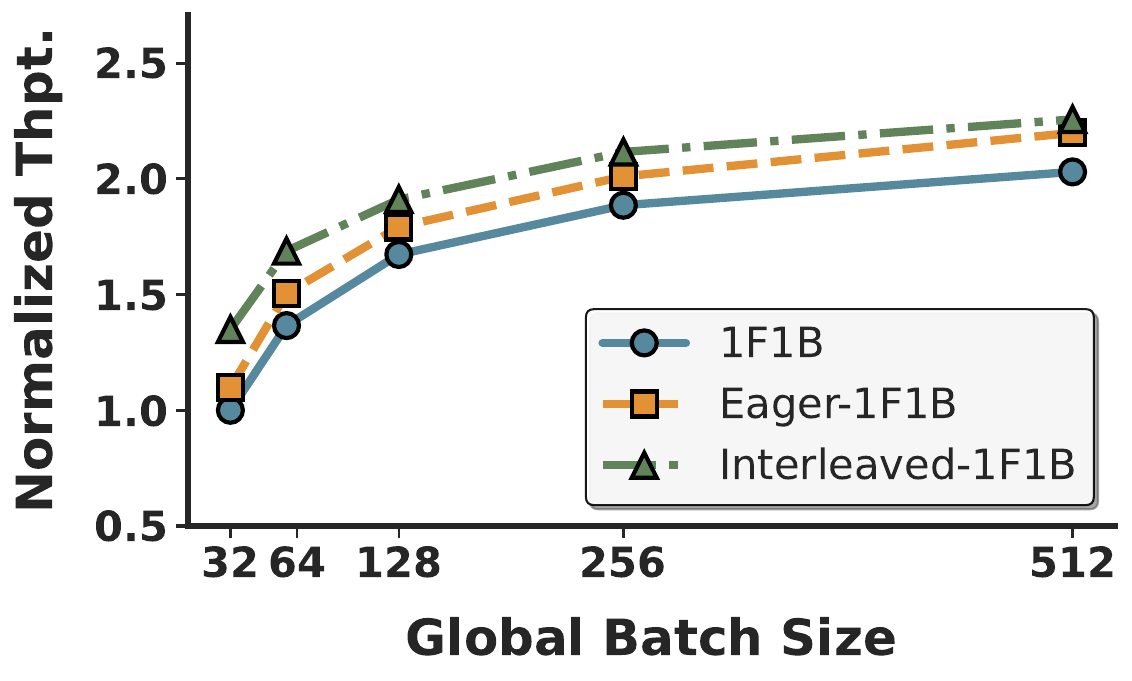}
		\vspace{-2em}
		\caption{Training throughput of Poseidon under different pipeline schedules.}
		\label{Fig.combined_throughput}
	\end{minipage}
	\vspace{-1em}
\end{figure}

\subsection{Effectiveness of the Search Algorithm}
\label{sec:effcetiveness-of-the-algorithm}

\subsubsection{Ablation studies for pruning strategies}
Poseidon’s pruning strategies are evaluated on LLaMA-30B under four scenarios: stage-level only, layer-level only, full pruning, and no pruning. Fig.~\ref{Fig.PruningEffec-gbs} shows the evolution of the best-found training time, yielding two insights:


\begin{itemize}[left=0pt, itemsep=0pt, labelsep=0.25em, label=--]
	\vspace{-0.25em}
	\item \textit{Effectiveness of individual pruning strategies.} Both stage-level and layer-level pruning accelerate convergence across all batch sizes, demonstrating that each strategy effectively filters out poor configurations.
	\vspace{-0.25em}
	\item \textit{Synergistic gains from combined pruning.} When combined, Poseidon achieves the fastest reduction in training time, highlighting the synergy: layer-level pruning guides the search toward high-quality regions, while stage-level pruning tightens bounds and prunes suboptimal branches.
\end{itemize}

\begin{figure}[t]
	\vspace{-0.8em}
	\centering
	\begin{minipage}[t]{0.49\linewidth}
		\centering
		\includegraphics[width=\linewidth]
		{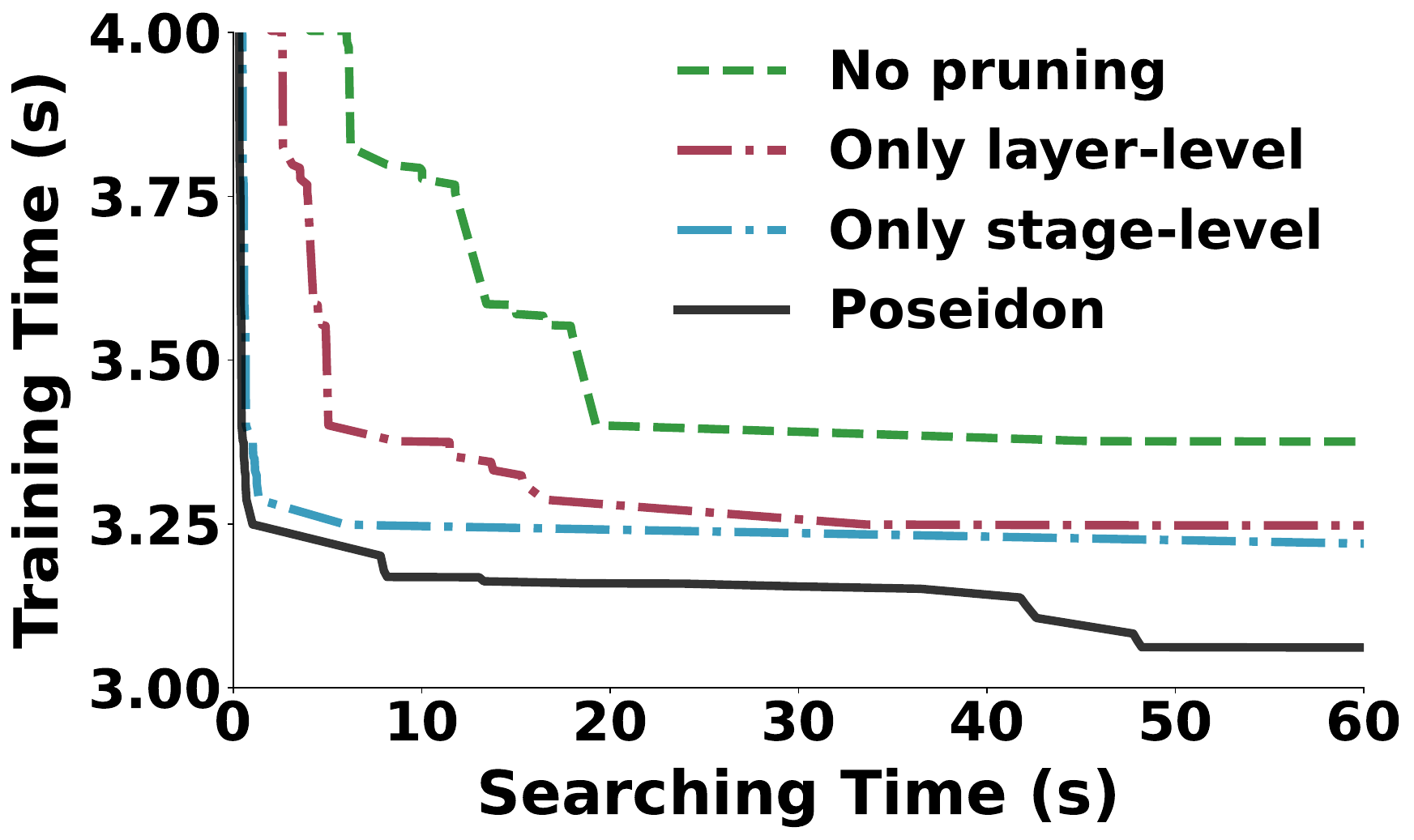}
		\vspace{-1.6em}
		\captionof{figure}{
			Breakdown analysis of Poseidon's pruning strategies.
		}
		\label{Fig.PruningEffec-gbs}
	\end{minipage}
	\hfill
	\begin{minipage}[t]{0.49\linewidth}
		\centering
		\includegraphics[width=\linewidth]
		{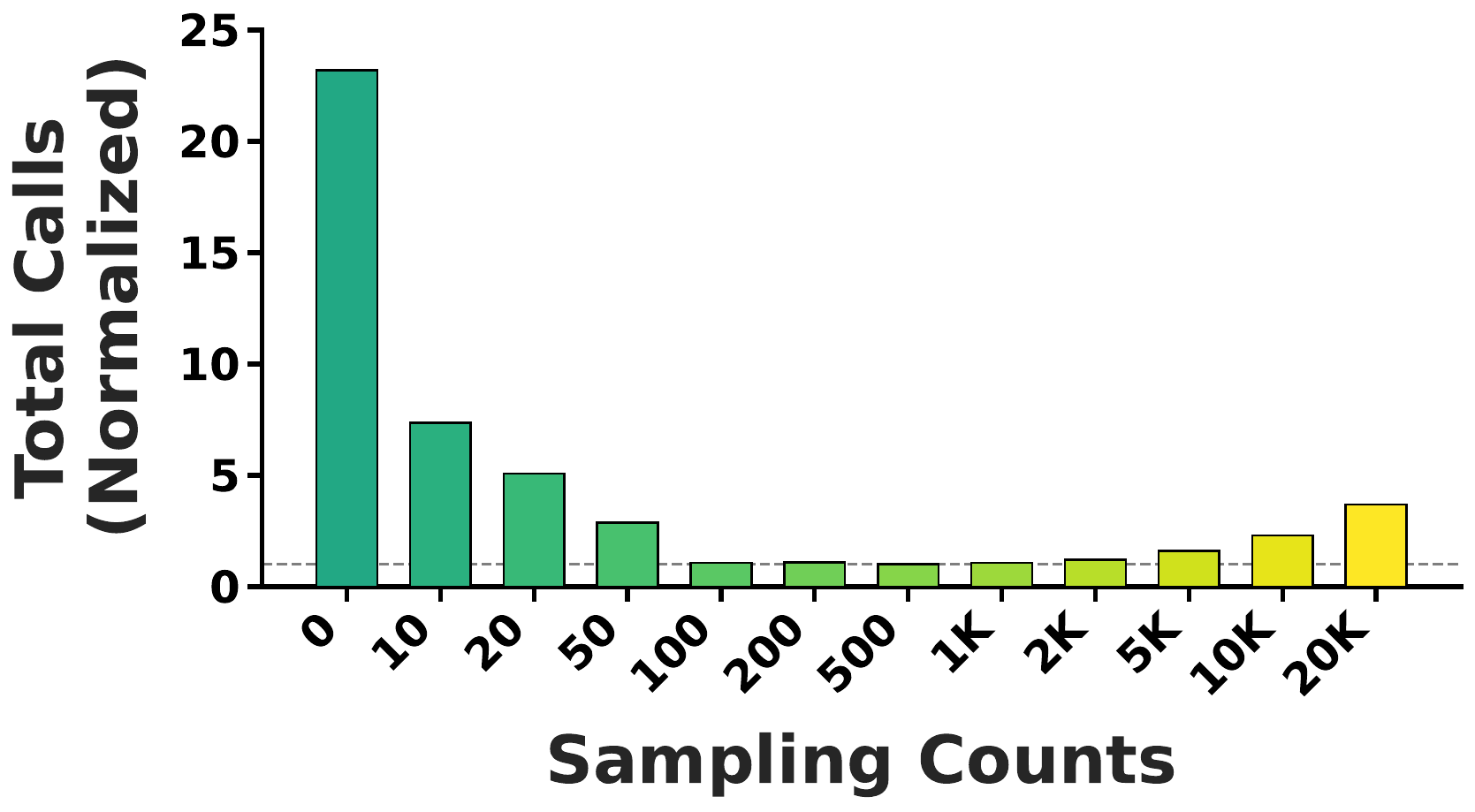}
		\vspace{-1.6em}
		\captionof{figure}{
			Impact of sampling frequency on search effectiveness.
		}
		\label{Fig.SampleTest}
	\end{minipage}
	\vspace{-1.5em}
\end{figure}

\subsubsection{Impact of warm-up sampling count}
\label{sec:impact-of-warm-up-sampling-count}
Poseidon adopts a warm-up sampling phase before stage-level pruning, in order to establish a high-quality initial baseline. A better baseline enables stricter pruning, which reduces simulator overhead.

We tested warm-up counts from 0 to 20,000 and found a trade-off: more samples improve pruning but increase upfront cost. Without sampling, weak initial configurations lead to 23.2$\times$ more simulator calls compared to the optimal warm-up sampling count, as shown in Fig. \ref{Fig.SampleTest}. Moderate sampling (e.g., 500 samples) quickly sets a competitive baseline, enabling efficient pruning. Thus, Poseidon defaults to 500 warm-up samples to balance cost and effectiveness.

\subsection{Evaluation in General Settings}

\subsubsection{Training MoE models}
The Mixture-of-Experts (MoE) architecture is increasingly popular for scaling LLMs, as adopted in Mixtral~\cite{jiang2024mixtral}, DeepSeek-V3~\cite{deepseekai2025deepseekv3technicalreport}, and Qwen3~\cite{yang2025qwen3technicalreport}. Although MoE models introduce multiple experts per layer, the profiling patterns remain similar to non-MoE models. Thus, the proposed training paradigm can be seamlessly applied to MoE architectures without modification.

We evaluated Poseidon on the Mixtral(8$\times$7B) model under Setting 4. As shown in Fig.~\ref{Fig.moe_exp}, Poseidon achieves 1.06$\times$ and 1.61$\times$ higher training throughput than Metis and HexiScale, respectively, while reducing search time by 52.07$\times$ and 3.15$\times$. 


\begin{figure}[t] 
	\vspace{-1em}
	\centering 
	\includegraphics[width=0.45\textwidth]{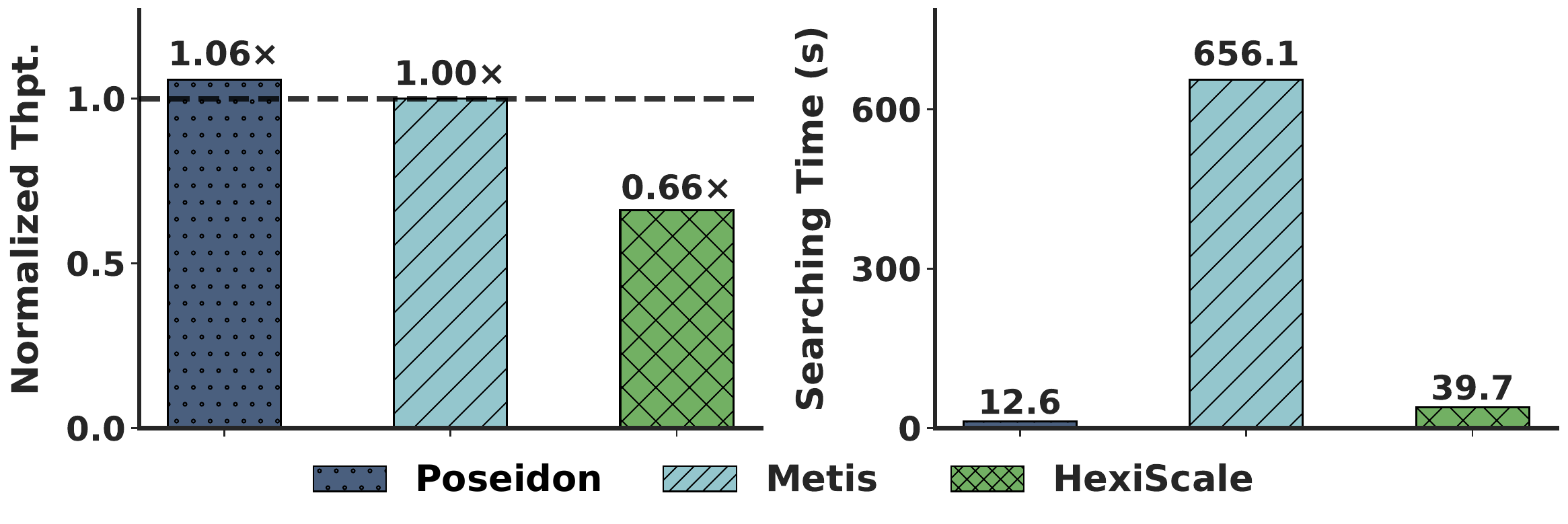}
	\vspace{-1em}
	\caption{Training throughput and search time of Poseidon compared to baseline systems on a MoE model.}
	\vspace{-0.25em}
	\label{Fig.moe_exp}
\end{figure}

\begin{figure}[t] 
	\centering 
	\vspace{-1.25em}
	\includegraphics[width=0.45\textwidth]{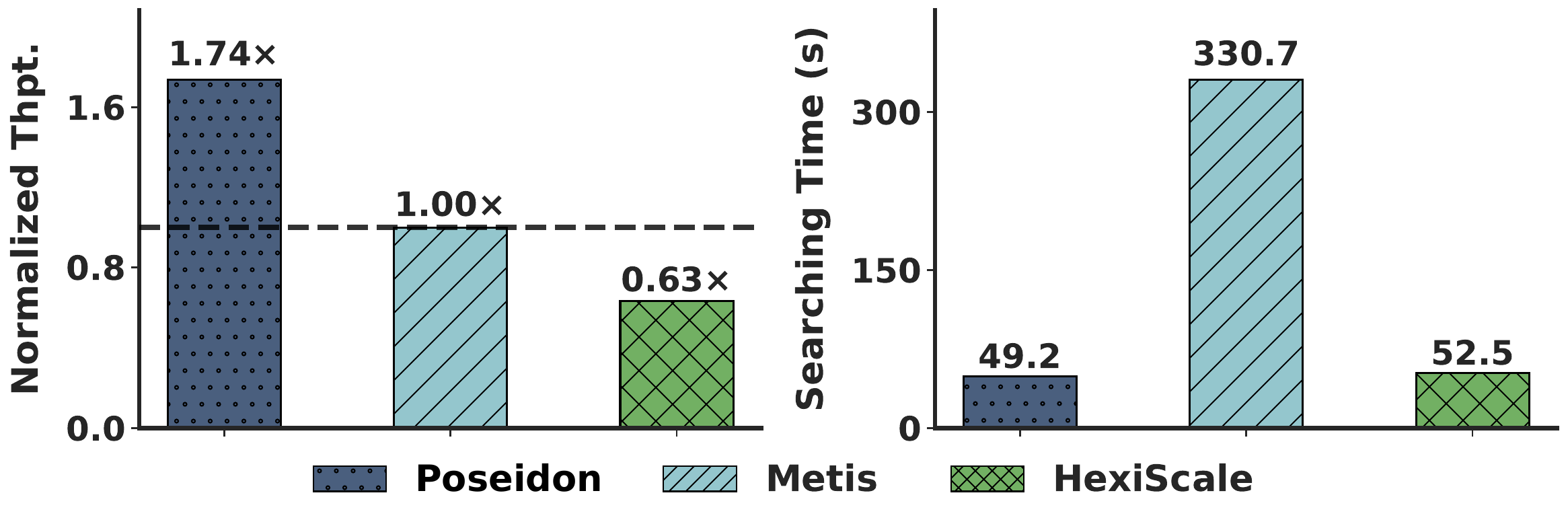}
	\vspace{-1.1em}
	\caption{The training throughput and searching time of Poseidon compared to baselines on heterogeneous GPU clusters. }
	\vspace{-1.5em}
	\label{Fig.gpu_exp}
\end{figure}

\subsubsection{LLM training across GPU Clusters}
Since the searching module and runtime system are loosely coupled, Poseidon's search module can be easily integrated with other runtime systems designed for different types of accelerators. 

To verify Poseidon's scalability in GPU environments, we integrated its search module into the HexiScale runtime and evaluated it under Setting 5. As shown in Fig.~\ref{Fig.gpu_exp}, Poseidon achieves 1.74$\times$ and 2.76$\times$ higher training throughput than Metis and HexiScale, respectively, while reducing search time by up to 6.72$\times$.

Notably, the GPU clusters exhibit greater compute and memory heterogeneity than the NPU clusters. The results show that Poseidon still delivers substantial throughput gains in such highly heterogeneous environments.

More generally, parallelism configuration search is independent of accelerator type, relying only on generic hardware profiles such as compute capability, memory capacity, and communication bandwidth. At this level of abstraction, there is no fundamental difference between GPUs and NPUs.

\subsubsection{Large-scale Simulation Experiments}
\label{sec:large-scale-simulation}
To further evaluate Poseidon at larger scales, we simulate its and baselines' performance under Setting 6 using LLaMA-3 (70B). All simulations use our simulator driven by the DAG-based cost model, emulating heterogeneous clusters with 96, 192, and 288 NPUs.

\textit{Superior scalability with increasing cluster size.} 
As shown in Fig.~\ref{Fig.largescale}, Poseidon consistently delivers the highest training throughput across all simulated scales in terms of performance. Specifically, it achieves a 1.50$\times$ to 1.58$\times$ improvement over Metis and a 1.45$\times$ to 1.55$\times$ improvement over HexiScale.

\textit{Efficient search at scale.} As shown in Table~\ref{table:large-scale-searching-time}, Poseidon substantially reduces search overhead, achieving 37.8$\times$–53.0$\times$ speedups over Metis and 4.5$\times$–5.3$\times$ over HexiScale. These results confirm that Poseidon’s theory-driven pruning strategies scale effectively, maintaining high efficiency despite exponential growth in device counts and search complexity.

\begin{figure}[t]
	\centering 
	\includegraphics[width=0.48\textwidth]{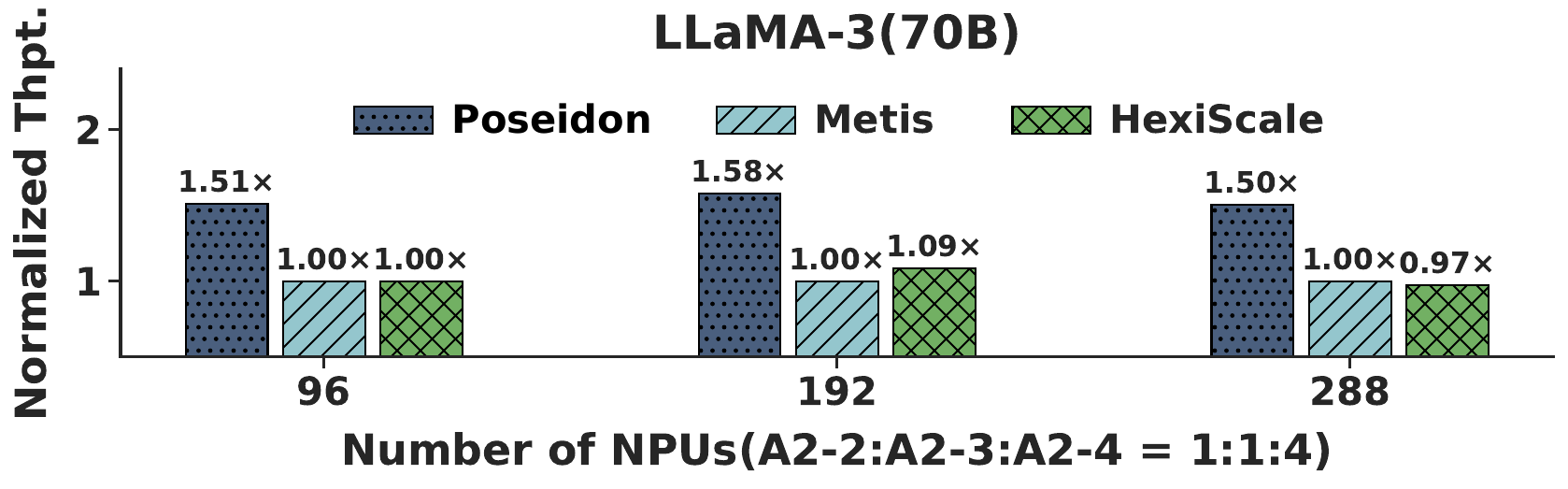}
	\vspace{-2em}
	\caption{Training throughput of Poseidon and baselines on heterogeneous NPU clusters at different scales.}
	\label{Fig.largescale}
	\vspace{-1em} 
\end{figure}

\begin{table}[t]
	\small
	\centering
	\caption{The search time of all systems on heterogeneous NPU clusters with varying scales}
	\vspace{-1em}
	\begin{tabular}{cccc} 
		\toprule 
		Number of NPUs  & 96  & 192  & 288     \\
		\midrule 
		Poseidon                    &       26.0s                 &   47.8s      & 75.8s \\ 
		Metis                          &       1378.7s             &   2382.6s & 2862.9s \\  
		HexiScale                  &     116.8s                 &   251.7s     & 394.7s \\ 
		\bottomrule 
	\end{tabular}	
	\label{table:large-scale-searching-time}
	\vspace{-1em}
\end{table}

\section{Related Work}

\textbf{Automating parallelism on heterogeneous devices.} Several recent studies have focused on automatically tuning parallelism for heterogeneous environments  \cite{um2024metis, yan2024flashflex}. Earlier work includes Piper \cite{tarnawski2021piper}, which first attempted to solve the problem of optimal partitioning via a two-level dynamic programming approach, and AMP \cite{li2022amp}, which defines a valid space of model-parallel strategies and underpins the cost models used by both Metis and HexiScale. Alpa \cite{zheng2022alpa} introduces a hierarchical search space and compiler-driven passes to derive efficient execution plans across parallelism levels—Metis itself builds on Alpa’s infrastructure. Galvatron \cite{miao2022galvatron} employs a decision-tree heuristic to prune the search space, followed by dynamic programming to identify optimal parallelism configurations. While recent work \cite{Rethinking} adopts a similar stage-level pruning strategy, it overlooks pipeline parallelism. Other contemporary approaches, such as Sailor \cite{sailor} and H2 \cite{tang2025h2towardsefficientlargescalellm}, leverage dynamic programming combined with heuristic pruning or depth-first search with heuristics, respectively, to prune the search space. However, the cost models that support their pruning mechanisms rely on static formulas similar to those used in Metis. In contrast, Poseidon is the first system to provide theoretical guarantees of optimal parallel configurations on heterogeneous clusters for LLM training. 

\textbf{DAG-based modeling for training systems.}
DAGs are a fundamental abstraction for representing dependency structures in deep neural network (DNN) and LLM training.
Prior work has leveraged DAGs for compilation optimization~\cite{moritz2018ray,wang2022overlap}, execution scheduling~\cite{tang2024fusionllm,hsia2024mad,jia2019beyond,jhoo2025pfeife}, communication optimization~\cite{wan2025coflow,wang2022overlap,won2023astra}, and formal reasoning~\cite{liang2025hapt}.
While these works demonstrate the expressive power of DAG abstractions, their objectives differ fundamentally from Poseidon. Notably, FlexFlow~\cite{jia2019beyond} is conceptually closest in its use of DAG modeling to Poseidon, it predicts parallel strategy performance via DAG simulation, but it focuses on operator-level parallel simulation for DNNs. Poseidon, by contrast, performs microbatch-level simulation tailored to LLM parallelism and explicitly supports large-scale parallel configuration search.  

\textbf{Pipeline schedule.}  Early pipeline parallelism schemes suffered significant device idle time ("pipeline bubbles"), limiting hardware utilization. While GPipe \cite{huang2019gpipe} introduced micro-batch pipelining to mitigate bubbles, it increases peak memory demands. Subsequent innovations like PipeDream's 1F1B schedule \cite{narayanan2019pipedream} reduced bubbles through interleaved forward/\allowbreak backward passes, inspiring variants including Eager 1F1B \cite{zhuang2023optimizing}, Interleaved 1F1B \cite{narayanan2021efficient}, and Seq1F1B \cite{sun2024seq1f1b}. Zero-bubble \cite{qi2024zero} achieved full stall elimination via backward-pass decomposition—later refined in \cite{qi2024pipeline} to reduce activation memory without reintroducing bubbles. Most recently, Zorse \cite{guo2025zorseoptimizingllmtraining} introduces interleaved GPipe-style pipeline scheduling with offloading parameters and activations to CPU memory. Mario \cite{liu2025mario} enables near-zero-cost activation checkpointing by overlapping recomputation with residual bubbles, enhancing both memory and compute efficiency. Poseidon's flexible modeling readily integrates all such variants. 

\textbf{Training frameworks and systems.} Numerous systems have been developed to support efficient large-scale model training \cite{zeng2025autoheteautomaticefficientheterogeneous, subramanya2023sia, chang2024frenzymemoryawareserverlessllm}. For example, Megatron-LM \cite{shoeybi2019megatron} pioneered model parallelism for multi-billion-parameter models; Whale \cite{jia2022whale} optimizes communication and memory efficiency on heterogeneous GPU clusters; and MegaScale \cite{jiang2024megascale} enables LLM training across over 10,000 GPUs by addressing scheduling, failure recovery, and network bottlenecks. 
These frameworks provide robust infrastructure for parallel training. To fully unleash the power of NPUs in our cluster, we built our system on top of MindSpeed-LLM, which offers similar functionalities with the above systems.

\vspace{-0.5em}
\section{Conclusion}
We present Poseidon, an automated parallelization system for LLM training on heterogeneous accelerator clusters. At its core, Poseidon designs a delicate DAG-based cost model to capture critical factors such as dependencies and overlap, enabling accurate training time estimation via critical-path analysis. Built upon this model, Poseidon incorporates two theoretically grounded searching strategies, i.e., stage- and layer-level pruning, to reduce the search space without compromising training efficiency. Moreover, Poseidon is empowered with strong generalizability to enhance efficiency for diverse paradigms like Eager 1F1B and Interleaved 1F1B. These innovations collectively improve Poseidon's efficiency, unlocking the scalable, efficient, and cost-effective LLM training in real-world heterogeneous environments.

\bibliographystyle{ACM-Reference-Format}
\bibliography{arXiv-Poseidon}

\clearpage
\appendix
\section*{Appendix}

In this appendix, we provide a detailed theoretical analysis of the stage-level pruning guarantee (Theorem~\ref{theorem.optimality-based pruning}) and the layer-level pruning guarantee (Theorem~\ref{theorem.Unimodal}).

\section{Stage-level Pruning Guarantee}
\label{sec:poseidon_stage_pruning_guarantee}
Recall that stage-level pruning operates on a partial configuration $S_x$ during the recursive search. The lower-bound function $g(S_x)$, defined in Eq.~\eqref{eq:stage-level-pruning}, characterizes the minimum possible training time that any complete configuration extending $S_x$ can achieve. Therefore, if this lower bound already exceeds the best observed time $T_{\min}$, continuing to explore the remaining stages cannot lead to a better configuration. The following proof formalizes this intuition and shows that the pruning rule in Theorem~\ref{theorem.optimality-based pruning} preserves optimality.

\begin{proof}[The proof of Theorem~\ref{theorem.optimality-based pruning}]
	We proceed by contradiction.
	Suppose there exists a full configuration $S^*$ extending $S_x$ such that its corresponding per-iteration training time $T(S^*) < T_{\min}$.
	
	Since $S^*$ extends $S_x$, the layer allocations for stages $1$ through $x$ remain unchanged, and only the assignments for stages $x+1$ onward may vary.
	
	By the definition of $g(S_x)$, which provides a valid lower bound on the total training time for any such extension, we must have:
	\[
	T(S^*) \geq g(S_x).
	\]
	
	However, we are assuming both $T(S^*) < T_{\min}$ and $g(S_x) > T_{\min}$, which leads to a contradiction:
	\[
	T(S^*) \geq g(S_x) > T_{\min} > T(S^*).
	\]
	
	Therefore, no extension of $S_x$ can yield a per-iteration training time smaller than $T_{\min}$, and all such configurations can be safely pruned.
\end{proof}

\section{Layer-level Pruning Guarantee}
\label{sec:poseidon_layer_pruning_guarantee}
The layer-level pruning guarantee is built upon the ridge-like property stated in Theorem~\ref{theorem.Unimodal}. Its proof is more involved but can be summarized in three steps. First, we identify a canonical structure of the critical path. Second, based on this characterization, we show that for stages assigned to device groups of the same type, any layer allocation can be transformed into a non-increasing ordering without increasing the critical-path cost, when temporarily ignoring the impact of device memory constraints. Finally, we combine this property with the fact that, under the 1F1B pipeline paradigm, activation storage requirements decrease monotonically with increasing stage indices \cite{liu2024aceso,jia2022whale, narayanan2019pipedream}. Together, these arguments establish the optimality of the ridge-like layer distribution used in $\mathsf{Neptune}$’s layer-level pruning.

\subsection{Preliminaries}

We consider a scenario of training LLMs on a pipeline with heterogeneous device groups. We introduce the following definitions and assumptions to formalize this process.

\begin{definition}
	Let $L$ denote the total number of layers in the model, $N$ the number of pipeline stages, and $M$ the number of micro-batches.
\end{definition}

\begin{definition}
	Let $l_i$ denote the number of layers assigned to stage $i$. Let $f_i$ and $b_i$ denote the per-layer forward and backward computation time at stage $i$, respectively. Then, let $F_i$ and $B_i$ denote the total forward and backward computation time of a micro-batch at stage $i$, respectively.
\end{definition}

Note that $f_i$ and $b_i$ are determined by the device group assigned to stage $i$. Additionally, we have:
\begin{equation}
	\begin{aligned}
		F_i &= l_i \cdot f_i, \\
		B_i &= l_i \cdot b_i.
	\end{aligned}
\end{equation}

\begin{definition}
	For any two stages $i$ and $j$, let $t_{i,j}$ denote the communication time required to transfer activations from stage $i$ to stage $j$.
\end{definition}

\begin{definition}
	Let $F_{i,j}$ and $B_{i,j}$ denote the forward and backward nodes, respectively, corresponding to the $i$-th pipeline stage and the $j$-th micro-batch in the computation DAG.
\end{definition}

Without loss of generality, we consider a single pipeline in the DAG for the purpose of this proof. Therefore, all DAGs referred to in the following analysis implicitly assume a single pipeline structure.

\begin{definition}
	In a DAG, each node has at most two outgoing edges. One edge represents the dependency within the same micro-batch across adjacent pipeline stages, referred to as the \textit{inter-stage edge}. The other edge represents the execution dependency within the same stage, referred to as the \textit{intra-stage edge}.
\end{definition}

As shown in Figure~\ref{Fig.PPgraph}, intra-stage edges are directed from left to right, while inter-stage edges of forward nodes are directed downward and those of backward nodes are directed upward.

\begin{figure*}[t] 
	\centering 
	\includegraphics[width=0.75\textwidth]{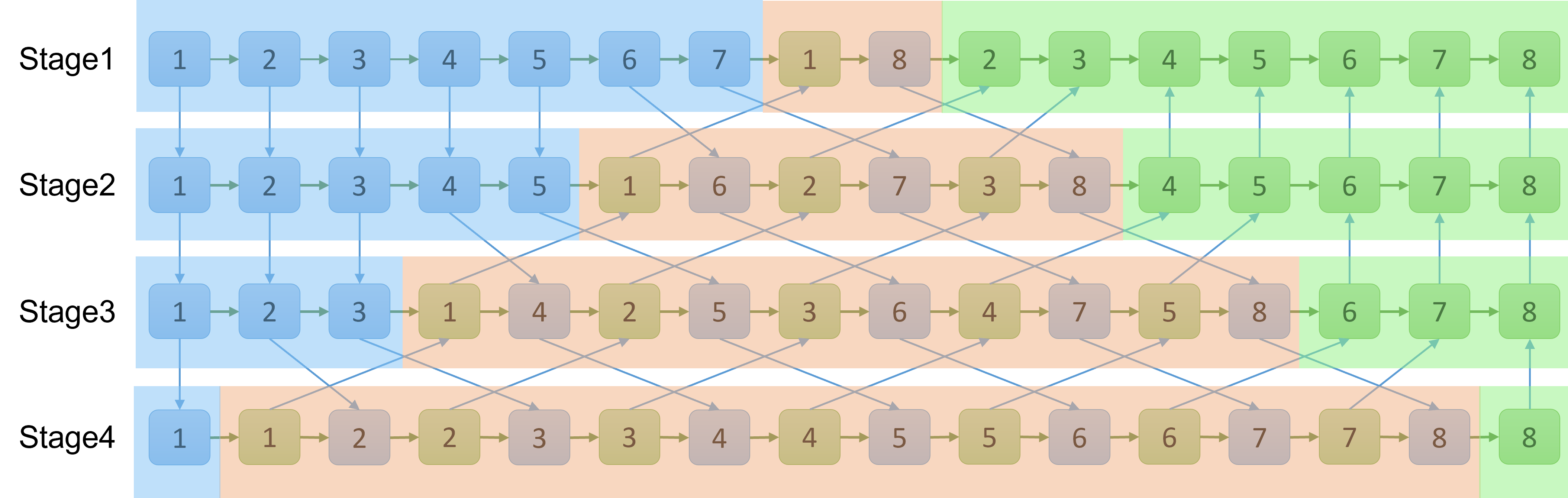} 
	\caption{Warmup (blue), steady (orange), and ending (green) phases in the DAG when $N = 4$ and $M = 8$.} 
	\label{Fig.PPgraph} 
\end{figure*}

\begin{figure*}[t] 
	\centering 
	\includegraphics[width=0.75\textwidth]{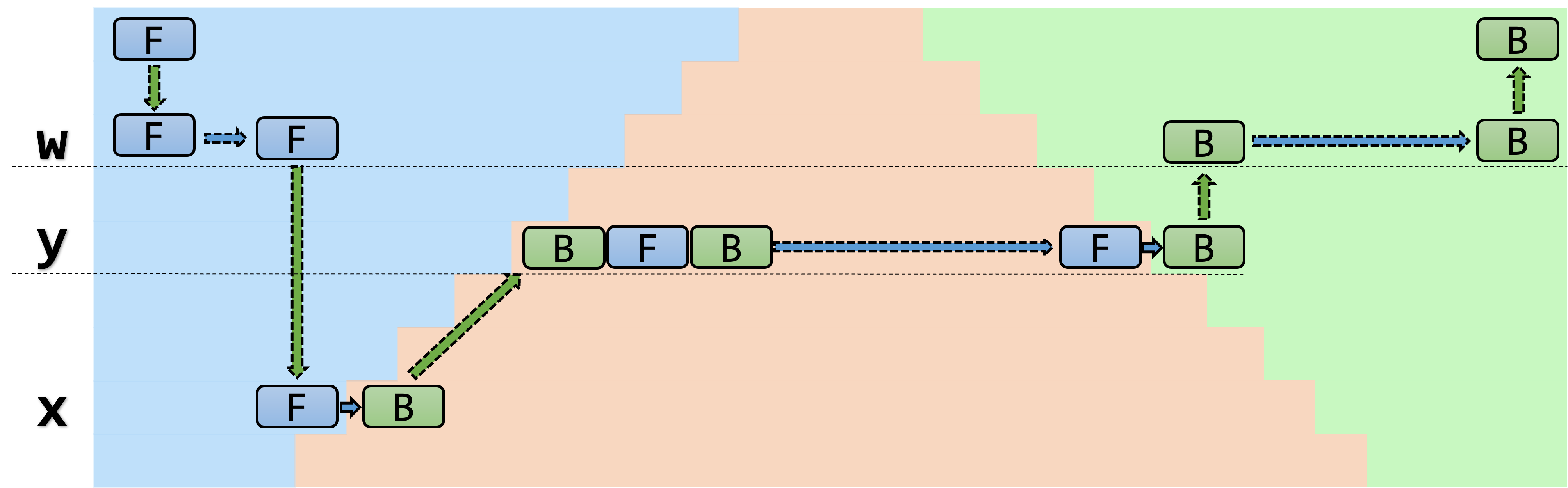} 
	\caption{A schematic representation of the path $K(x,y,w)$} 
	\label{Fig.Kxyw} 
\end{figure*}

\begin{definition}
	In a DAG, the \textit{critical path} is the longest weighted path, defined as the path with the maximum sum of weights over all its nodes and edges.
\end{definition}

In addition, we introduce four assumptions to support the subsequent theoretical analyses.

\begin{assumption}
	\label{assump.BFratio}
	For any $i$ $(1 \le i \le N)$, $b_i = r \cdot f_i$, where $r > 1$ is the ratio of backward to forward computation time of a single layer.
\end{assumption}

Backward propagation is typically more time-consuming than forward propagation due to gradient computation \cite{qi2024zero}.

\begin{assumption}
	\label{assump.ComuniTime}
	For any communication time $t_{i,j}$ $(1 \le i, j \le N)$ and any stage $k$ $(1 \le k \le N)$, it holds that
	\[
	t_{i,j} \le \frac{r-1}{2} F_k.
	\]
\end{assumption}

\begin{assumption}
	It holds that $M \ge 2N$.
\end{assumption}

In practice, the number of micro-batches is typically much larger than the number of pipeline stages \cite{liu2024aceso}, making this assumption reasonable.

\begin{assumption}
	\label{assump.DifferBound}
	Let $F_{\max} = \max_{1 \le i \le N} F_i$ and $F_{\min} = \min_{1 \le i \le N} F_i$. We assume that
	\[
	F_{\min} \ge \alpha \cdot F_{\max},
	\]
	where $\alpha = \max\left(\frac{2}{3}, \frac{r}{1+r}\right)$.
\end{assumption}

In this proof, we focus on the Eager-1F1B pipeline scheduling strategy; other scheduling variants can be analyzed similarly. Under Eager-1F1B, for the $i$-th stage $(1 \le i \le N)$, the execution consists of three phases:
(i) a warmup phase with $2(N-i)+1$ forward micro-batches;
(ii) a steady phase consisting of $M - 2(N-i) - 1$ groups of alternating backward-and-forward micro-batches;
(iii) a cooldown (ending) phase with $2(N-i)+1$ backward micro-batches.
We refer to these three parts as the \textit{warmup phase}, \textit{steady phase}, and \textit{ending phase}, respectively. Figure~\ref{Fig.PPgraph} presents an example.

\subsection{Critical Path Characterization}
\label{sec:critiacal_path_characterization}

In this part, we identify a canonical structure of the critical path in the DAG and present its formal definition as follows.
\begin{definition}
	\label{def.Kxyw}
	In a DAG, let $K(x,y,w)$ denote a path constructed as follows:
	\begin{itemize}[left=2pt,itemsep=0pt,topsep=0pt,parsep=0pt,partopsep=0pt,nosep,labelsep=0.3em]
		\item start from $F_{1,1}$;
		\item move downward until reaching stage $w$;
		\item move right for $2(N-x)$ steps;
		\item move downward until reaching stage $x$;
		\item move right by one step into the steady phase;
		\item move upward until reaching stage $y$;
		\item move right until entering the ending phase;
		\item move upward until reaching stage $w$;
		\item move right until reaching $B_{w,M}$;
		\item finally, move upward until reaching $B_{1,M}$,
	\end{itemize}
	where $1 \le w \le y \le x \le N$, with $y = \arg\max_{i \in [y,x]} F_i$ and $w = \arg\max_{i \in [1,x]} F_i$.
\end{definition}
Figure~\ref{Fig.Kxyw} provides an example of $K(x,y,w)$. Next, we prove the existence of the critical path $K(x,y,w)$. Before that, we introduce several definitions and lemmas.

\begin{figure*}[t] 
	\centering 
	\includegraphics[width=0.75\textwidth]{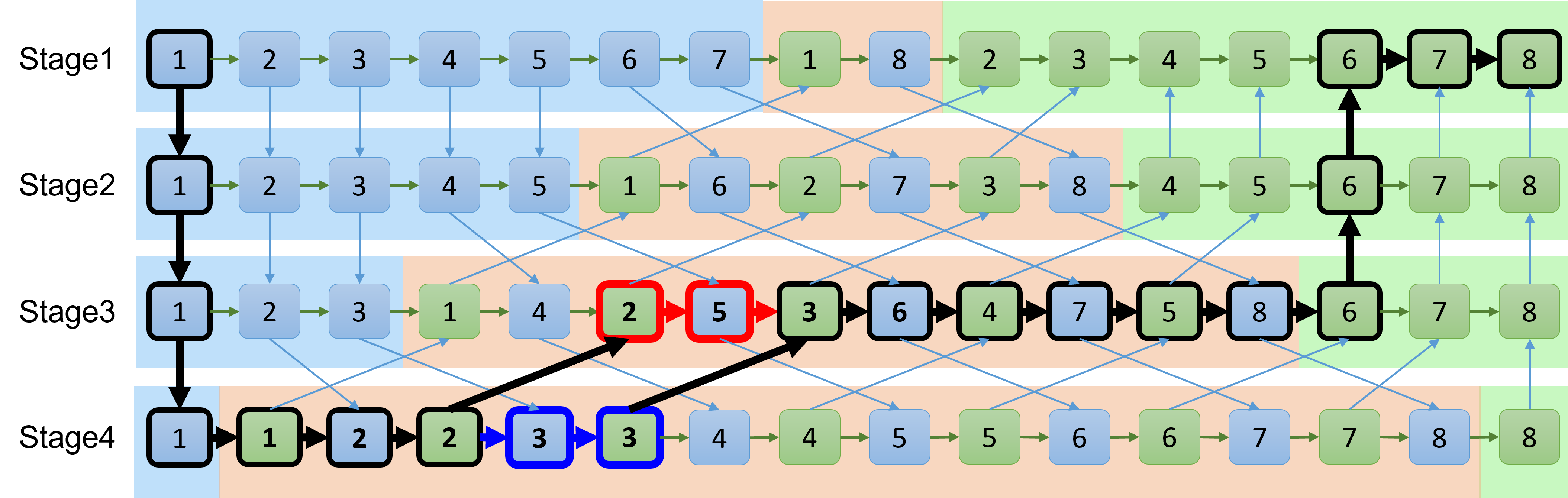}
	\caption{The step transfer operation during steady phase, transfering $2$ rightward steps from one stage to another stage.}
	\label{Fig.StepTransfer}
\end{figure*}
\begin{figure*}[t] 
	\centering 
	\includegraphics[width=0.75\textwidth]{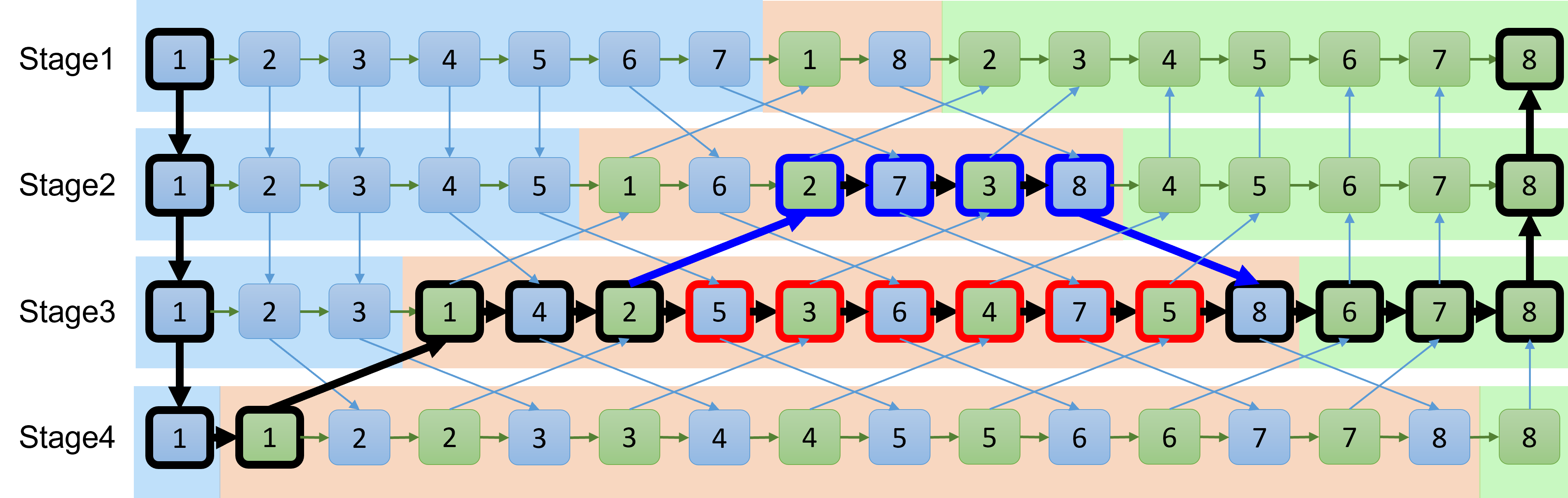}
	\caption{An example of Peak Flattening. Replacing a peak(blue) with a rightward segment(red).}
	\label{Fig.PeakFlatten}
\end{figure*}

\begin{definition}
	The following terms describe canonical path structures in a DAG, which will be used in the subsequent lemmas.
	\begin{itemize}[left=2pt,itemsep=0pt,topsep=0pt,parsep=0pt,partopsep=0pt,nosep,labelsep=0.3em]
		\item \textit{Segment} is a path consisting of consecutive steps in a single direction, which can be upward, downward, or rightward.
		\item \textit{Plateau} is a segment consisting of at least $2$ consecutive rightward steps in the steady phase.
		\item \textit{Peak} consists of three consecutive segments in the steady phase: first, $h$ consecutive upward steps; second, $j$ consecutive rightward steps; and third, $h$ consecutive downward steps, where $h \ge 0$, $j \ge 1$, and $h$ is referred to as the height of the peak. The \textit{peak-top} is the segment formed by the $j$ rightward steps.
		\item \textit{Valley} consists of three consecutive segments in the steady phase: first, $d$ consecutive downward steps; second, $j$ consecutive rightward steps; and third, $d$ consecutive upward steps, where $d \ge 0$, $j \ge 1$, and $d$ is referred to as the depth of the valley. The \textit{valley-bottom} is the segment formed by the $j$ rightward steps.
	\end{itemize}
\end{definition}

	%
	%

\begin{lemma}
	\label{lemma.onlyplateau}
	If there exists a critical path containing plateaus in the steady phase, then there exists a critical path such that all its plateaus in the steady phase are contained in a single stage, denoted as the $x^*$-th stage. Moreover, let $S$ denote the set of stages visited by the path in the steady phase. Then it holds that
	\vspace{-0.5em}
	\[
	F_{x^*} = \max_{i \in S} F_i.
	\]
\end{lemma}

\begin{proof}
	We prove the lemma via a plateau relocation (exchange) argument.
	
	Consider a critical path in which the steady phase contains plateaus located in at least two different stages. Let these stages be $x$ and $y$ with $1 \le x < y \le N$, and assume without loss of generality that $F_x \le F_y$. Let $\sigma_x$ denote a plateau in stage $x$ consisting of $k \ge 2$ consecutive rightward steps.
	
	We construct a new path by transferring two rightward steps from $\sigma_x$ to a plateau in stage $y$. Figure~\ref{Fig.StepTransfer} illustrates this operation. After this operation, the modified segment in stage $x$ becomes a shorter plateau $\sigma'_x$ with $(k-2)$ steps; if $k=2$, it disappears. 
	
	Next, we show that this operation does not decrease the total path cost. Since each rightward step in a plateau contributes one forward and one backward computation, the local cost change satisfies
	\begin{equation}
		F_x + B_x = (1+r)F_x \le (1+r)F_y = F_y + B_y,
	\end{equation}
	where the inequality follows from $F_x \le F_y$.
	
	Therefore, the modified path is no worse than the original critical path. In the case $F_x < F_y$, the operation strictly improves the path cost, contradicting the optimality of the original critical path. Hence, we must have $F_x = F_y$.
	
	By repeatedly applying this relocation argument, all plateaus can be merged into a single stage $x^*$ such that $F_{x^*} = \max_{i \in S} F_i$, where $S$ is the set of stages visited in the steady phase. This completes the proof.
\end{proof}

\begin{lemma}
	\label{lemma.onlypeak}
	If a critical path contains peaks or valleys in the steady phase, then there exists a critical path whose peak-tops and valley-bottoms are all located at stages with the maximum forward computation time among the stages traversed in the steady phase. That is, for any such stage $x$,
	\vspace{-0.5em}
	\[
	F_x = \max_{i \in S} F_i,
	\]
	\vspace{-0.5em}
	where $S$ denotes the set of stages traversed by this path in the steady phase.
\end{lemma}

\begin{figure*}[tbp]
	\centering
	\begin{subfigure}[b]{0.46\textwidth}
		\centering
		\includegraphics[width=\textwidth]{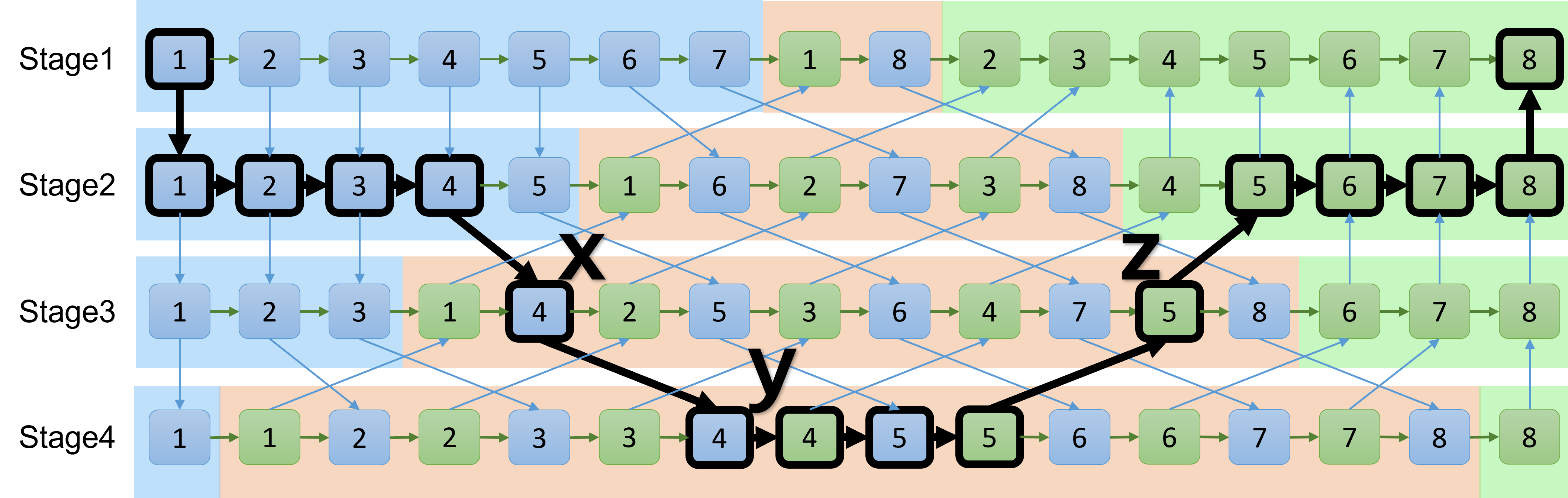}
		\caption{}
		\label{Fig.xyzValley}
	\end{subfigure}
	\hfill
	\begin{subfigure}[b]{0.46\textwidth}
		\centering
		\includegraphics[width=\textwidth]{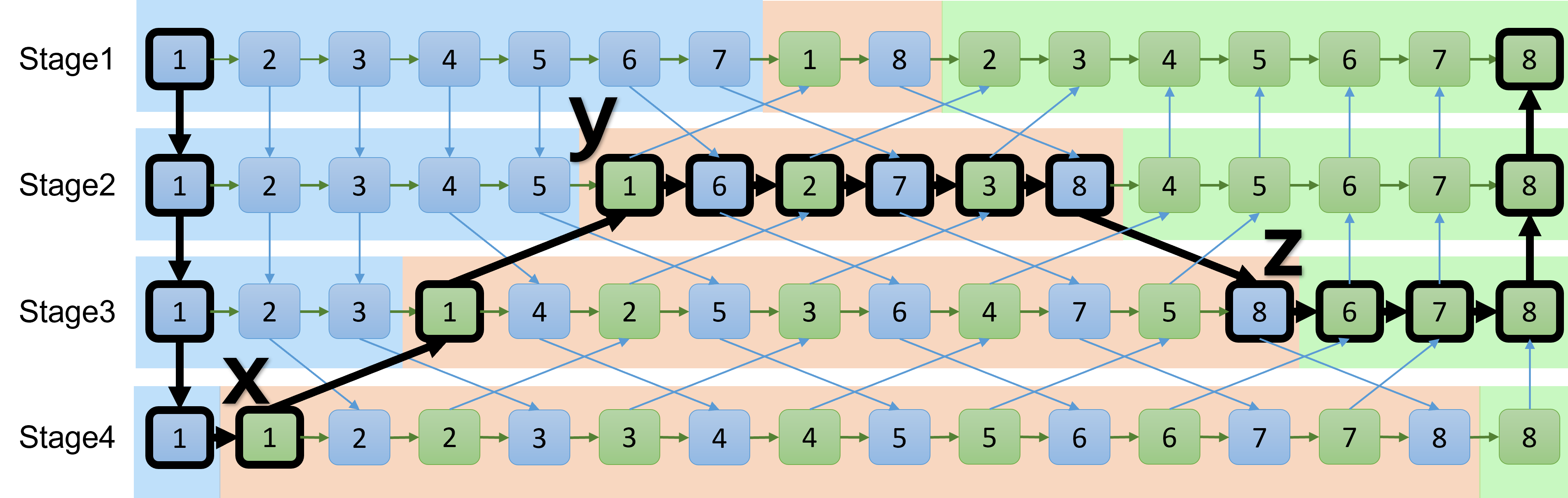}
		\caption{}
		\label{Fig.xyzPeak}
	\end{subfigure}
	
	\vspace{0.5em}
	
	\begin{subfigure}[b]{0.46\textwidth}
		\centering
		\includegraphics[width=\textwidth]{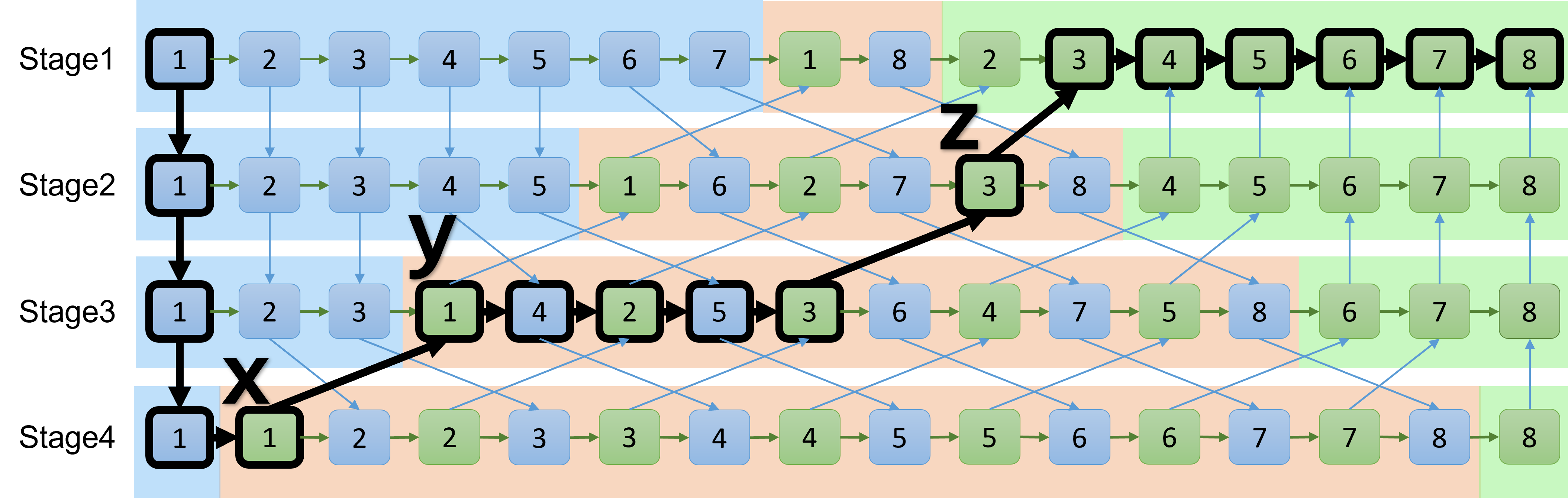}
		\caption{}
		\label{Fig.xyzZigZag}
	\end{subfigure}
	\hfill
	\begin{subfigure}[b]{0.46\textwidth}
		\centering
		\includegraphics[width=\textwidth]{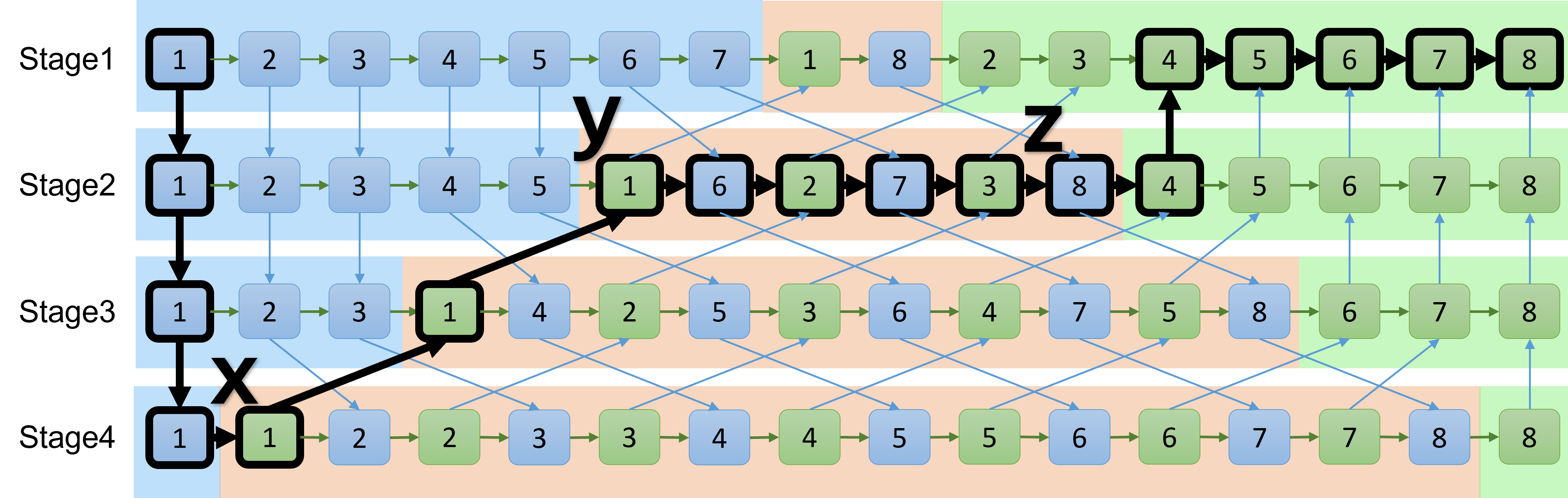}
		\caption{}
		\label{Fig.xyzFinal}
	\end{subfigure}
	
	\caption{Examples of $\mathcal{P}(x,y,z)$.}
	\label{Fig.XYZ}
\end{figure*}

\begin{proof}
	We prove the lemma by contradiction. We first consider the case of a peak; the case of a valley can be proved symmetrically by filling the valley-bottom.
	
	Suppose there exists a critical path whose steady phase contains a peak with its peak-top located at stage $x$. Let $S$ denote the set of stages traversed by this path in the steady phase. Assume, for contradiction, that there exists a stage $y \in S$ such that
	\begin{equation}
		F_y > F_x .
	\end{equation}
	Without loss of generality, consider the case $y > x$. Let $h = y-x$ be the height between stages $x$ and $y$, and let the peak-top contain $c \ge 1$ rightward steps. Then the peak subpath first moves upward from stage $y$ to stage $x$, then moves right along the peak-top, and finally moves downward from stage $x$ back to stage $y$.
	
	We replace this peak subpath with a flat rightward segment at stage $y$, consisting of $c+4h$ rightward steps. This operation is referred to as \textit{cutting the peak-top}. Figure~\ref{Fig.PeakFlatten} illustrates this replacement.
	
	We now compare the costs of the two subpaths. The cost of the original peak subpath can be written as
	\begin{equation}
		C_{\mathrm{peak}}
		=
		\sum_{i=x}^{y-1}
		\left(F_i+B_i+t_{i,i+1}+t_{i+1,i}\right)
		+
		\frac{c-1}{2}(F_x+B_x).
	\end{equation}
	By Assumption~\ref{assump.BFratio}, we have $B_i=rF_i$. Since $F_y>F_x$ and $F_y$ is no smaller than the stages involved in this replacement, we have $F_i \le F_y$ for $i \in [x,y]$. Moreover, by Assumption~\ref{assump.ComuniTime},
	\begin{equation}
		t_{i,i+1}+t_{i+1,i} \le (r-1)F_y .
	\end{equation}
	Therefore,
	\begin{equation}
		\begin{aligned}
			C_{\mathrm{peak}}
			&\le
			\sum_{i=x}^{y-1}
			\left((1+r)F_y+(r-1)F_y\right)
			+
			\frac{c-1}{2}(1+r)F_x \\
			&=
			2rhF_y+\frac{c-1}{2}(1+r)F_x .
		\end{aligned}
	\end{equation}
	
	After cutting the peak-top, the replacement subpath at stage $y$ has cost
	\begin{equation}
		C_{\mathrm{flat}}
		=
		\left(\frac{c-1}{2}+2h\right)(F_y+B_y)
		=
		\left(\frac{c-1}{2}+2h\right)(1+r)F_y .
	\end{equation}
	Since $r>1$ and $F_y>F_x$, we obtain
	\begin{equation}
		C_{\mathrm{flat}} > C_{\mathrm{peak}} .
	\end{equation}
	Thus, replacing the original peak subpath with the flat segment yields a strictly longer path, contradicting the assumption that the original path is critical.
	
	Therefore, the peak-top of a critical path cannot be located at a stage whose forward computation time is smaller than that of another stage traversed in the steady phase. That is, the peak-top must lie in a stage $x$ satisfying
	\begin{equation}
		F_x = \max_{i \in S} F_i .
	\end{equation}
	
	The same argument applies to a valley by reversing the vertical directions, which corresponds to filling the valley-bottom. Hence, every peak-top and valley-bottom can be placed at a stage with the maximum forward computation time among stages in $S$.
	
	If multiple stages attain the same maximum value, the above replacement can be applied without decreasing the path cost, allowing all peak-tops and valley-bottoms to be moved to a single such stage $x^*$. Therefore, there exists a critical path in which all peak-tops and valley-bottoms lie in the same stage $x^*$, with
	\begin{equation}
		F_{x^*}=\max_{i\in S}F_i .
	\end{equation}
	This completes the proof.
\end{proof}

\begin{theorem}
	\label{theorem.Pxyz}
	Consider a critical path that enters the steady phase at stage $x$ and exits the steady phase at stage $z$. There exists a critical path $\mathcal{P}$ with the same entry and exit stages, and a stage $y$, such that the following properties hold:
	\begin{itemize}[left=2pt,itemsep=0pt,topsep=0pt,parsep=0pt,partopsep=0pt,nosep,labelsep=0.3em]
		\item In the steady phase, $\mathcal{P}$ consists of three segments:
		\begin{enumerate}[left=2pt,itemsep=0pt,topsep=0pt,parsep=0pt,partopsep=0pt,nosep,labelsep=0.3em]
			\item starting from the entry at stage $x$, it moves continuously upward or downward to stage $y$;
			\item it then moves continuously rightward along stage $y$;
			\item from stage $y$, it moves continuously upward or downward to the exit at stage $z$.
		\end{enumerate}
		\item Stage $y$ has the maximum forward computation time among all stages traversed by $\mathcal{P}$ in the steady phase, i.e.,
		\vspace{-0.5em}
		\[
		F_y =
		\max_{\min(x,y,z) \le i \le \max(x,y,z)} F_i .
		\vspace{-0.5em}
		\]
	\end{itemize}
\end{theorem}

Figure~\ref{Fig.XYZ} presents four examples of the critical-path structures characterized in Theorem~\ref{theorem.Pxyz}.

\begin{proof}
	By Lemma~\ref{lemma.onlyplateau}, there exists a critical path in which all plateaus in the steady phase are located in a single stage that attains the maximum forward computation time among the stages traversed in the steady phase. By Lemma~\ref{lemma.onlypeak}, all peak-tops and valley-bottoms can also be placed in such a stage without decreasing the path cost.
	
	Therefore, there exists a critical path whose all rightward movements in the steady phase are concentrated in a single stage, denoted as stage $y$, where
	\begin{equation}
		F_y =
		\max_{\min(x,y,z) \le i \le \max(x,y,z)} F_i .
	\end{equation}
	After entering the steady phase at stage $x$, the path must first move vertically to stage $y$, then move rightward along stage $y$, and finally move vertically from stage $y$ to the exit stage $z$. Hence, the steady-phase portion of $\mathcal{P}$ consists of the three segments stated above. This completes the proof.
\end{proof}

\begin{definition}
	\label{def.Pxyz}
	Let $\mathcal{P}(x,y,z)$ denote the critical path characterized in Theorem~\ref{theorem.Pxyz}.
\end{definition}

After characterizing the critical path in the steady phase, we next analyze its structure in the warmup and ending phases.

\begin{lemma}
	\label{lemma.ending}
	Let $\mathcal{P}$ be a critical path that enters the ending phase at stage $z$. The subpath of $\mathcal{P}$ in the ending phase, denoted by $\mathcal{P}_{\text{end}}$, can be adjusted such that it consists of three segments:
	\begin{itemize}[left=2pt,itemsep=0pt,topsep=0pt,parsep=0pt,partopsep=0pt,nosep,labelsep=0.3em]
		\item starting from stage $z$, move continuously upward to stage $w$, where
		\[
		F_w = \max_{1 \le i \le z} F_i ;
		\]
		\item move continuously rightward along stage $w$;
		\item from stage $w$, move continuously upward to stage $1$.
	\end{itemize}
\end{lemma}

\begin{proof}
	This lemma can also be proved by relocating rightward steps. If $\mathcal{P}_{\text{end}}$ contains rightward steps in stages other than $w$, these steps can be replaced by rightward steps in stage $w$ without decreasing the cost of $\mathcal{P}$, as illustrated in Figure~\ref{Fig.endingPhase}. By repeatedly applying this relocation operation until all rightward steps lie in stage $w$, we obtain the desired critical path.
\end{proof}

\begin{figure}[t] 
	\centering 
	\includegraphics[width=0.25\textwidth]{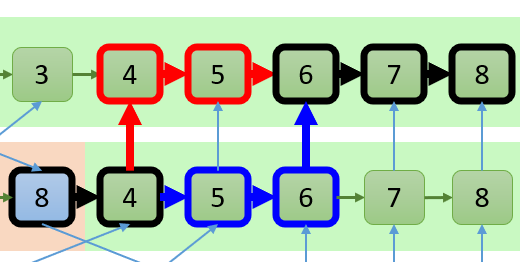}
	\caption{Relocating rightward steps from one stage to another stage(from blue steps to red steps). The cost change depends on the backward computation time difference of the two stages}
	\label{Fig.endingPhase}
\end{figure}

Moreover, lemma \ref{lemma.ending} has a symmetric lemma for warmup phase. That is, for any critical path $\mathcal{P}$, all of the rightward steps in warmup phase can be relocated to a stage $m$, where $F_m = \max\limits_{1\le i\le x}F_i$ and $x$ is the enter stage of $\mathcal{P}$. It can be proved similarly.

\begin{lemma}
	\label{lemma.ending_replace}
	Consider a critical path whose ending-phase subpath moves rightward at stage $w$. Suppose that, after the last rightward segment in the steady phase, the path moves continuously upward from stage $y$ to stage $w$. Then
	\vspace{-0.5em}
	\[
	B_w \ge F_y + B_y .
	\]
\end{lemma}

\begin{proof}
	We prove this lemma by contradiction. Suppose that $B_w < F_y+B_y$. Then one backward computation at stage $w$ in the ending phase can be replaced by one forward computation and one backward computation at stage $y$ in the steady phase, as illustrated in Figure~\ref{Fig.endingReplace}. This replacement preserves the dependency constraints of the DAG and strictly increases the path cost. This contradicts the assumption that the original path is critical. Therefore, we must have $B_w \ge F_y+B_y$.
\end{proof}

\begin{figure}[t] 
	\centering 
	\includegraphics[width=0.3\textwidth]{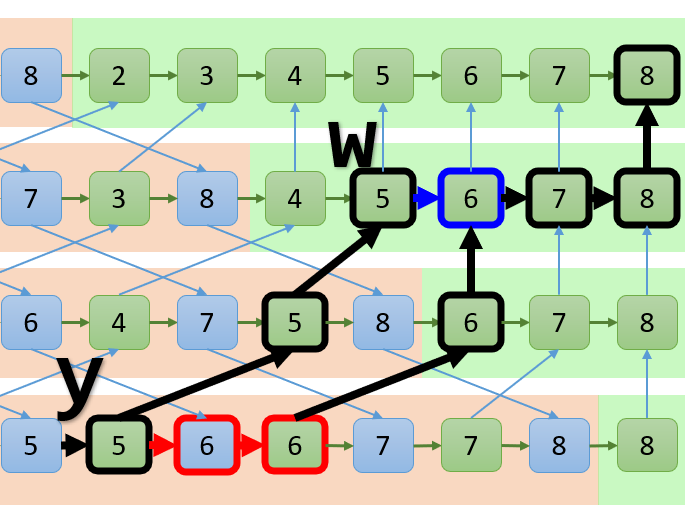}
	\caption{Delaying the position of the upward step, the cost change depends on the difference between a Backward time(blue) in stage $w$ and a Forward+Backward time(red) in stage $y$.}
	\label{Fig.endingReplace}
\end{figure}

\begin{theorem}
	\label{theorem.yleqxz}
	There exists a critical path $\mathcal{P}(x,y,z)$, as defined in Definition~\ref{def.Pxyz}, such that
	\[
	y \le x
	\quad \text{and} \quad
	y \le z .
	\]
	Moreover, stage $y$ satisfies
	\[
	F_y \ge \max_{i \in [y,\max(x,z)]} F_i .
	\]
\end{theorem}

\begin{figure*}[t] 
	\centering 
	\includegraphics[width=0.75\textwidth]{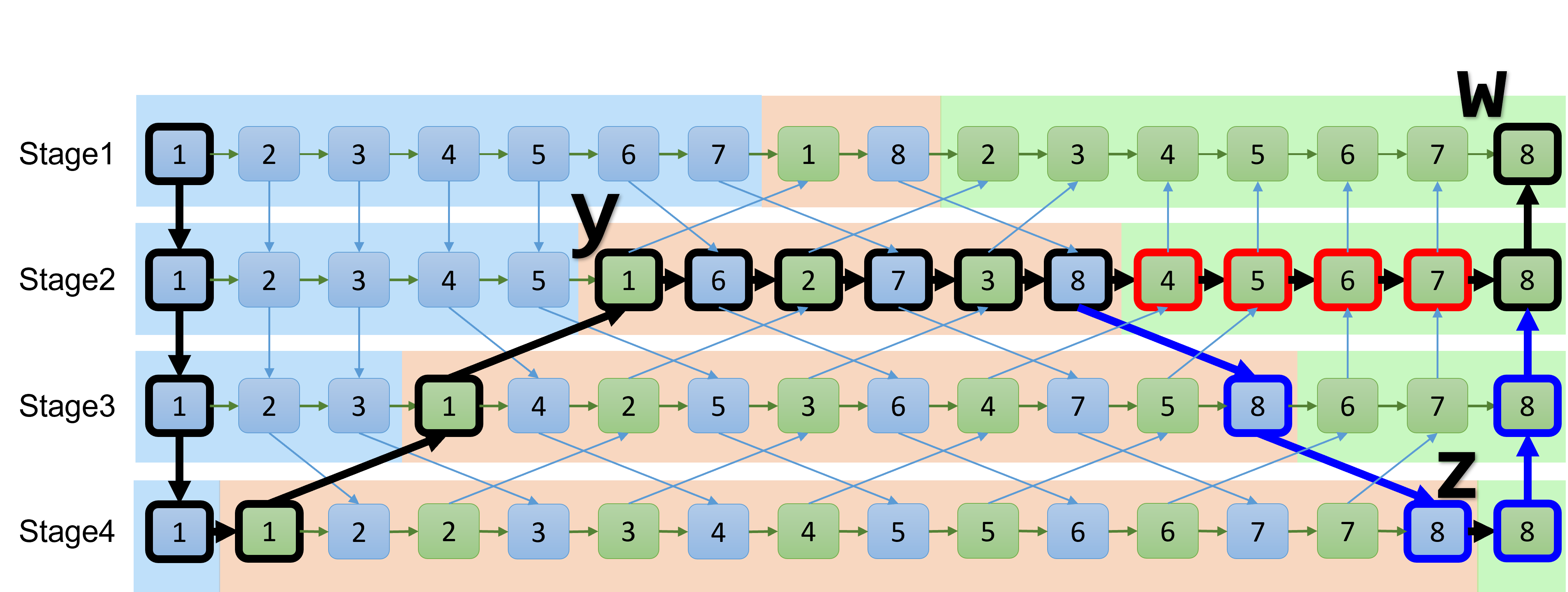}
	\caption{Comparing the difference between the path $\mathcal{P}(x,y,z)$ and the path going directly from stage $y$ to ending phase, where $z>y$.} 
	\label{Fig.ZeqY}
\end{figure*}

\begin{proof}
	We first prove that there exists a critical path satisfying $y \le z$. Suppose, for contradiction, that a critical path $\mathcal{P}(x,y,z)$ satisfies $y>z$.
	
	Since $y>z$, after the last rightward segment in the steady phase, the path must move continuously upward from stage $y$ to stage $z$. After entering the ending phase at stage $z$, by Lemma~\ref{lemma.ending}, the ending-phase subpath can be adjusted such that its rightward steps are located at a stage $w$, where
	\begin{equation}
		F_w = \max_{1 \le i \le z} F_i .
	\end{equation}
	Thus, the path continues moving upward from stage $z$ to stage $w$ before taking rightward steps in the ending phase.
	
	By Lemma~\ref{lemma.ending_replace}, we have
	\begin{equation}
		B_w \ge F_y + B_y .
	\end{equation}
	Therefore, replacing one forward-backward pair at stage $y$ in the steady phase by one backward computation at stage $w$ does not decrease the path cost. We can repeatedly apply this replacement until the path can no longer exit the steady phase at stage $z$.

	Next, we further replace one forward-backward pair at stage $y$ by one forward computation and one backward computation at stage $w$. Since $F_w>0$, we have
	\begin{equation}
		F_w+B_w > F_y+B_y .
	\end{equation}
	Hence, this replacement strictly increases the path cost while preserving the DAG dependencies, as illustrated in Figure~\ref{Fig.endingReplace}. This contradicts the assumption that $\mathcal{P}(x,y,z)$ is a critical path. Therefore, the case $y>z$ is impossible, and there exists a critical path satisfying
	\begin{equation}
		y \le z .
	\end{equation}
	
	It remains to prove that there exists a critical path satisfying $y \le x$. This part follows by a symmetric argument on the warmup phase. Specifically, if $y>x$, then before reaching the rightward segment at stage $y$ in the steady phase, the path must move continuously downward from stage $x$ to stage $y$. By applying the analogous replacement argument in the warmup phase, one can construct a path with strictly larger cost, again contradicting the criticality of $\mathcal{P}(x,y,z)$. Therefore, the case $y>x$ is also impossible, and we have
	\begin{equation}
		y \le x .
	\end{equation}
	
	Combining the two results, there exists a critical path $\mathcal{P}(x,y,z)$ such that $y\le x$ and $y\le z$. Moreover, by Theorem~\ref{theorem.Pxyz}, stage $y$ has the maximum forward computation time among the stages traversed by the steady-phase subpath. Since $y\le x$ and $y\le z$, this gives
	\begin{equation}
		F_y \ge \max_{i \in [y,\max(x,z)]} F_i .
	\end{equation}
	This completes the proof.
\end{proof}

Furthermore, we show that there exists a critical path $\mathcal{P}(x,y,z)$ with $z=y$, as illustrated in Figure~\ref{Fig.xyzFinal}.
\begin{theorem}
	\label{theorem.yeqz}
	There exists a critical path $\mathcal{P}(x,y,z)$ such that
	\vspace{-0.5em}
	\[
	z = y \le x .
	\vspace{-0.5em}
	\]
	Moreover, it satisfies
	\[
	F_y \ge \max_{i \in [y,x]} F_i .
	\]
\end{theorem}

\begin{proof}
	We prove the theorem by contradiction. By Theorem~\ref{theorem.yleqxz}, there exists a critical path $\mathcal{P}(x,y,z)$ such that $y\le x$ and $y\le z$. Suppose, for contradiction, that no such critical path satisfies $z=y$. Then every critical path of this form must have $z>y$.
	
	Consider such a critical path $\mathcal{P}(x,y,z)$ with $z>y$. By Theorem~\ref{theorem.yleqxz}, we have
	\begin{equation}
		F_y \ge \max_{i\in [y,z]} F_i .
	\end{equation}
	By Lemma~\ref{lemma.ending}, the ending-phase subpath can be adjusted such that its rightward steps are located at a stage $w$, where
	\begin{equation}
		F_w = \max_{1\le i\le z}F_i .
	\end{equation}
	We choose $w$ as the smallest stage satisfying the above equality. Since stage $y$ is also no smaller than all stages in $[y,z]$ in terms of forward computation time, we have $w\le y$ and $F_w\ge F_y$.

\begin{figure*}[t] 
	\centering 
	\includegraphics[width=0.75\textwidth]{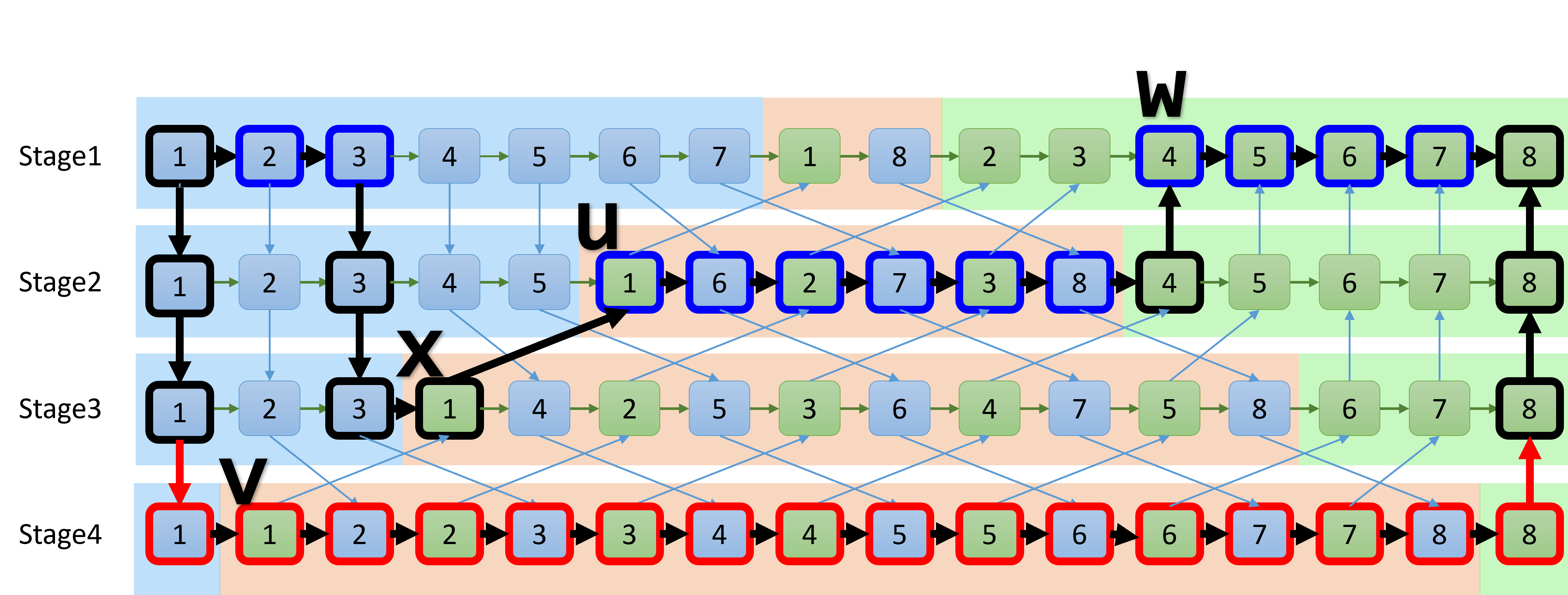}
	\caption{The example of comparing the difference between the path $K(x,u,w)$(blue) and the path $K(v,v,w)$(red) when $N=4,u=2,v=4,w=1,x=3$.}
	\label{Fig.compare}
\end{figure*}	
	
	Now consider the portion of $\mathcal{P}(x,y,z)$ from the last rightward segment in the steady phase to the rightward segment in the ending phase. Since $z>y$, this portion first moves downward from stage $y$ to stage $z$, enters the ending phase, and then moves upward from stage $z$ to stage $w$. We replace this subpath with one that exits the steady phase directly from stage $y$, as shown in Figure~\ref{Fig.ZeqY}. This replacement changes the exit stage of the steady phase from $z$ to $y$.
	
	The gain of this replacement is
	\begin{equation}
		\sum_{i=y+1}^{z}
		\left(
		2B_w -
		\left(F_i+B_i+t_{i-1,i}+t_{i,i-1}\right)
		\right).
	\end{equation}
	Since $F_w\ge F_y\ge F_i$ for all $i\in[y,z]$, and by Assumption~\ref{assump.ComuniTime}, the communication cost is bounded by the computation cost. Therefore, each term in the summation is non-negative, and the replacement does not decrease the total path cost.
	
	Thus, we obtain another critical path whose steady-phase exit stage becomes $z=y$, contradicting the assumption that no critical path of this form satisfies $z=y$. Therefore, there exists a critical path $\mathcal{P}(x,y,z)$ such that $z=y\le x$.

	Finally, since $z=y$, the maximality property in Theorem~\ref{theorem.yleqxz} reduces to
	\begin{equation}
		F_y \ge \max_{i\in[y,x]}F_i .
	\end{equation}
	This completes the proof.

\end{proof}

Therefore, we obtain the following theorem.

\begin{theorem}
	\label{theorem.ExisKxyw}
	There exist $x$, $y$, and $w$ such that $K(x,y,w)$ is a critical path.
\end{theorem}

\begin{proof}
	By Theorem~\ref{theorem.yeqz} and Lemma~\ref{lemma.ending}, there exists a critical path whose steady-phase and ending-phase structures satisfy the properties stated in Definition~\ref{def.Kxyw}. Therefore, this critical path can be represented as $K(x,y,w)$ for some $x$, $y$, and $w$.
\end{proof}

Figure~\ref{Fig.Kxyw} illustrates the structure of $K(x,y,w)$ defined in Definition~\ref{def.Kxyw}. Theorem~\ref{theorem.ExisKxyw} reduces the critical path to a concise canonical form, which facilitates the subsequent analysis.

\subsection{Ridge-like Allocation Guarantee}
\label{sec:ridge_like_allocation_guarantee}

Based on the canonical critical-path structure identified above, we next establish the ridge-like allocation guarantee for layer-level pruning.

\begin{definition}
	Let $C(\mathcal{P})$ denote the cost of a path $\mathcal{P}$ in the DAG.
\end{definition}

\begin{theorem}
	\label{theorem.exclusion}
	For any pair of stages $u$ and $v$ satisfying $1 \le u < v \le N$, if $F_u \le F_v$, then there exist $x$, $y$, and $w$, with $y \neq u$, such that $K(x,y,w)$ is a critical path.
\end{theorem}

\begin{proof}
	Consider any pair of stages $u$ and $v$ such that $1 \le u < v \le N$ and $F_u \le F_v$. Suppose that $K(x,u,w)$ is a critical path. It suffices to show that
	\begin{equation}
		\label{inequal.cost}
		C\left(K(x,u,w)\right) \le C\left(K(v,v,w)\right).
	\end{equation}
	Indeed, if Eq.~\eqref{inequal.cost} holds, then $K(v,v,w)$ has a cost no smaller than that of the critical path $K(x,u,w)$, and hence $K(v,v,w)$ is also a critical path. Since its middle index is $v \neq u$, this gives the desired critical path. Figure~\ref{Fig.compare} illustrates this comparison.
	
	We prove Eq.~\eqref{inequal.cost} through the following steps.
	
	\begin{enumerate}[left=2pt,itemsep=0pt,topsep=0pt,parsep=0pt,partopsep=0pt,nosep,labelsep=0.3em]
		\item From Definition~\ref{def.Kxyw}, the cost of $K(x,y,w)$ can be written as
		\begin{align}
			C(K(x,y,w))
			&= \underbrace{
				\left(\sum_{i=1}^{x-1}
				\left(F_i+B_i+t_{i,i+1}+t_{i+1,i}\right)\right)
				+F_x+B_x
			}_{\text{downward and upward steps}} \label{formula.cost} \\
			&\quad + \underbrace{
				2(N-x)F_w + 2(N-y)B_w
			}_{\text{rightward steps in stage } w}  \notag \\
		    &\quad + \underbrace{
				\left(M-2(N-y)-1\right)(F_y+B_y)
			}_{\text{rightward steps in stage } y}. \notag
		\end{align}
		
		\item To prove Eq.~\eqref{inequal.cost}, we first use the following equivalent form:
	\begin{align}\label{inequal.diff}
		\begin{split}
			& C\left(K(x,u,w)\right) \leq C\left(K(v,v,w)\right)\\
			\iff & C\left(K(v,v,w)\right) - C\left(K(x,u,w)\right) \geq 0.
		\end{split}
	\end{align}
		Substituting Eq.~\eqref{formula.cost} and $B_i=rF_i$ for $1\le i\le N$ into Eq.~\eqref{inequal.diff}, we obtain the following equivalent condition:
		\begin{align}
			\label{inequal.diffDetail}
			\begin{split}
				0
				&\le
				\left(
				\sum_{i=x+1}^{v}
				\left((1+r)F_i+t_{i-1,i}+t_{i,i-1}\right)
				\right)
				-2(v-x)F_w \\
				&\quad
				+\left(M-2(N-u)-1\right)(1+r)(F_v-F_u) \\
				&\quad
				+2(v-u)\left((1+r)F_v-rF_w\right).
			\end{split}
		\end{align}
		Therefore, proving Eq.~\eqref{inequal.diffDetail} is equivalent to proving Eq.~\eqref{inequal.cost}.
		
		\item Next, consider the following sufficient condition:
		\begin{align}
			\label{inequal.B}
			\begin{split}
				0
				&\le
				\alpha(v-x)(1+r)F_w
				-2(v-x)F_w \\
				&\quad
				+2(v-u)\left((1+r)F_v-rF_w\right).
			\end{split}
		\end{align}
		We have Eq.~\eqref{inequal.B} $\implies$ Eq.~\eqref{inequal.diffDetail}, because the right-hand side of Eq.~\eqref{inequal.B} is no larger than the right-hand side of Eq.~\eqref{inequal.diffDetail}. This follows from the following facts:
		\begin{itemize}[left=2pt,itemsep=0pt,topsep=0pt,parsep=0pt,partopsep=0pt,nosep,labelsep=0.3em]
			\item $F_v-F_u \ge 0$;
			\item $F_i \ge \alpha F_w$ for all $i\in [x+1,v]$;
			\item $t_{i,i-1}+t_{i-1,i} \ge 0$ for all $i\in [x+1,v]$.
		\end{itemize}
		
		\item By rearranging the terms in Eq.~\eqref{inequal.B}, moving the terms involving $F_w$ to the left-hand side and the terms involving $F_v$ to the right-hand side, and considering the sign of the coefficient, Eq.~\eqref{inequal.B} can be equivalently written as
		\begin{align}
			\label{inequal.C}
			F_w
			\le
			\left(
			\frac{
				2(v-u)(1+r)
			}{
				2(v-x)+2r(v-u)-\alpha(1+r)(v-x)
			}
			\right)F_v .
		\end{align}
		Thus, Eq.~\eqref{inequal.C} is equivalent to Eq.~\eqref{inequal.B}.
		
		\item Since $u\le x < v$, it can be verified that when
		\[
		\alpha=\max\left(\frac{2}{3},\frac{r}{1+r}\right),
		\]
		the following inequality holds:
		\begin{equation}
			\label{Formu.Key}
			\frac{1}{\alpha}
			\le
			\frac{
				2(v-u)(1+r)
			}{
				2(v-x)+2r(v-u)-\alpha(1+r)(v-x)
			}.
		\end{equation}
		Moreover, by Assumption~\ref{assump.DifferBound}, we have
		\[
		F_w \le \frac{1}{\alpha}F_v .
		\]
		Therefore, Eq.~\eqref{Formu.Key} implies Eq.~\eqref{inequal.C}.
		
		\item Finally, Assumption~\ref{assump.DifferBound} gives
		\begin{equation}
			\label{equal.assm}
			\alpha = \max\left(\frac{2}{3},\frac{r}{1+r}\right).
		\end{equation}
	\end{enumerate}
	
	Combining the above implications, we have
	\[
	\eqref{equal.assm}
	\implies
	\eqref{Formu.Key}
	\implies
	\eqref{inequal.C}
	\iff
	\eqref{inequal.B}
	\implies
	\eqref{inequal.diffDetail}
	\iff
	\eqref{inequal.cost}.
	\]
	Therefore, Eq.~\eqref{inequal.cost} holds, and the theorem follows.
\end{proof}

Theorem~\ref{theorem.exclusion} shows that when there exist two stages $u$ and $v$ with $F_u \le F_v$, the case $y=u$ can be safely excluded from the analysis of the critical path $K(x,y,w)$.

Before proving Theorem~\ref{theorem.Unimodal}, we present two supporting lemmas. For any layer allocation scheme $l_1,l_2,\cdots,l_N$, we focus on the stages $s_1,s_2,\cdots,s_k$ that are assigned to device groups of the same type, and consider their corresponding layer sequence $l_{s_1},l_{s_2},\cdots,l_{s_k}$.

%
\begin{lemma}
	\label{lemma.swap}
	Given the layer sequence $l_{s_1},l_{s_2},\cdots,l_{s_k}$, if device memory constraints are temporarily ignored and the sequence is not monotonically non-increasing, then the following adjustment procedure can transform it into a monotonically non-increasing sequence without increasing the cost of the critical path in the DAG.
\end{lemma}
\noindent \textit{Adjustment procedure.}
\begin{enumerate}[left=2pt,itemsep=0pt,topsep=0pt,parsep=0pt,partopsep=0pt,nosep,labelsep=0.3em]
	\item Identify an index $i$ $(1\le i<k)$ such that $l_{s_i}<l_{s_{i+1}}$. If no such index exists, terminate.
	\item Reallocate one layer from stage $s_{i+1}$ to stage $s_i$:
	\vspace{-0.25em}
	\[
	l_{s_i} \leftarrow l_{s_i}+1, 
	\qquad
	l_{s_{i+1}} \leftarrow l_{s_{i+1}}-1 .
	\]
	\item Repeat from Step 1.
\end{enumerate}

\begin{proof}
	We show that the layer reallocation in Step 2 does not increase the cost of the critical path. There are two cases.
	
	\begin{itemize}[left=2pt,itemsep=0pt,topsep=0pt,parsep=0pt,partopsep=0pt,nosep,labelsep=0.3em]
		\item If $l_{s_i}\le l_{s_{i+1}}$ still holds after the adjustment, then by Theorem~\ref{theorem.exclusion}, the case $y=s_i$ can be excluded from the analysis of the critical path $K(x,y,w)$. Therefore, the cost of the critical path does not increase.
		
		\item If $l_{s_i}>l_{s_{i+1}}$ holds after the adjustment, then the initial difference between the two stages must be exactly one. In this case, by direct calculation, $C(K(x,s_i,w))$ in the DAG after the adjustment is no greater than $C(K(x,s_{i+1},w))$ in the DAG before the adjustment. Hence, the cost of the critical path still does not increase.
	\end{itemize}
	
	Therefore, each adjustment step preserves or reduces the cost of the critical path. Repeating this procedure eventually transforms the sequence into a monotonically non-increasing sequence, which proves the lemma.
\end{proof}

\begin{lemma}
	\label{lemma.memlimit}
	Let $l_i^m$ denote the memory-constrained layer limit of stage $i$ $(1\le i\le N)$. Then the sequence 
	$l^m_{s_1},l^m_{s_2},\cdots,l^m_{s_k}$ is monotonically non-decreasing.
\end{lemma}

\begin{proof}
	Under the 1F1B pipeline paradigm, activation storage requirements decrease monotonically as the stage index increases. Since the considered stages are assigned to device groups of the same type and thus have identical memory capacities, their memory-constrained layer limits form a monotonically non-decreasing sequence.
\end{proof}
\vspace{-0.25em}
Finally, we prove Theorem~\ref{theorem.Unimodal}.
\vspace{-0.25em}
\begin{proof}[The proof of Theorem~\ref{theorem.Unimodal}]
	Lemma~\ref{lemma.swap} shows that, when device memory constraints are ignored, there exists an optimal layer allocation whose layer sequence over device groups of the same type is monotonically non-increasing. Lemma~\ref{lemma.memlimit} further shows that the corresponding memory-constrained layer limits form a monotonically non-decreasing sequence.
	
	Combining these two monotonicity properties, the feasible optimal allocation must first follow the non-increasing tendency induced by the critical-path cost, while also respecting the non-decreasing upper bounds imposed by memory constraints. Therefore, the optimal layer distribution over device groups of the same type exhibits a unimodal structure. This completes the proof.
\end{proof}


\end{document}